\documentclass[12pt]{article}

\usepackage[margin=1in]{geometry}

\usepackage{mathrsfs, mathtools}
\usepackage{latexsym, amssymb, amsmath, amsthm}
\usepackage{bm}
\usepackage{slashed}

\usepackage[round]{natbib}

\usepackage[T1]{fontenc}
\usepackage[utf8]{inputenc}

\usepackage{setspace}
\usepackage[shortlabels]{enumitem}
\usepackage[multiple, flushmargin]{footmisc}
\usepackage[title,toc,titletoc]{appendix}
\usepackage{titlesec}
\titleformat{\section}[block]{\large\normalfont\filcenter}{\thesection.}{.5em}{}
\titleformat{\subsection}[block]{\normalfont\bfseries\filright}{\thesubsection.}{.5em}{}

\usepackage{graphicx}
\usepackage{caption,subcaption}
\usepackage{comment}

\usepackage[dvipsnames]{xcolor}

\usepackage{tikz}
\usetikzlibrary{positioning}

\usepackage{pgfplots}
\pgfplotsset{compat=1.14}
\usetikzlibrary{patterns, decorations.pathreplacing, arrows}

\tikzset{
  every node/.style = {
    text centered,
    line width = .5pt,
    anchor = center,
  },
  every label/.style = {
    fill = white, anchor = mid,
  },
  every path/.style = {
    > = stealth
  },
  point/.style args = {(#1)#2}{
    rounded corners,
    fill = white,
    minimum height = 20pt,
    minimum width = 10pt,
    label = { [name = #1] above:#2 },
  },
  point a/.style args = {(#1)#2}{
    rounded corners,
    fill = white,
    minimum height = 10pt,
    minimum width = 20pt,
    label = { [name = #1] right:#2 },
  },
  point b/.style args = {(#1)#2}{
    rounded corners,
    fill = white,
    minimum height = 10pt,
    minimum width = 75pt,
    label = { [name = #1] right:#2 },
  },
}

\usepackage{longtable,booktabs,array}
\newcolumntype{L}[1]{>{\raggedright\arraybackslash}p{#1}}

\usepackage{etoolbox}

\newcommand{\rank}{\operatorname{r}}
\newcommand{\priority}{\pi}

\newcommand{\ch}{\operatorname{Ch}}
\newcommand{\nw}{\operatorname{NW}}
\newcommand{\front}{\operatorname{F}}

 \newcommand{\norm}[1]{\left\lVert#1\right\rVert_1}
\newcommand{\abs}[1]{\left| #1 \right|}

\newtheorem{theorem}{Theorem}

\newtheorem{lemma}{Lemma}
\newtheorem{proposition}{Proposition}
\newtheorem{claim}{Claim}

\theoremstyle{definition}
\newtheorem{definition}{Definition}

\newtheorem{axiom}{Axiom}

\newtheorem{example}{Example}

\usepackage[
    colorlinks=true,
    linkcolor=blue,
    citecolor=blue,
    urlcolor=blue,
    filecolor=blue,
    hyperfootnotes=false,
    hyperindex,
    breaklinks
]{hyperref}

\usepackage{cleveref}
\usepackage{quoting}
\renewcommand{\emptyset}{\varnothing}

\usepackage{lmodern}

\allowdisplaybreaks
\begin{document}

\title{Distributional Preferences for Market Design\thanks{We are grateful to Junya Murayama for excellent research assistance, and to Battal Do\u{g}an and Koji Yokote for helpful discussions.}}
\author{Federico Echenique\thanks{Department of Economics, UC Berkeley. Contact: \href{mailto:fede@econ.berkeley.edu}{fede@econ.berkeley.edu}} \hspace*{1.25em} Teddy Mekonnen\thanks{Department of Economics, Brown University. Contact: \href{mailto:mekonnen@brown.edu}{mekonnen@brown.edu}}\hspace*{1.25em} M. Bumin Yenmez\thanks{Department of Economics, Washington University in St.~Louis, Durham University (UK), and \"Ozye\u{g}in University, Turkey. Contact: \href{mailto:bumin@wustl.edu}{bumin@wustl.edu}.}}

\maketitle

\begin{abstract}
Institutions care both about whom they select and the composition of the resulting group. We identify an upper-bound property and two exchange properties on distributional preferences. Conditional on the upper-bound property, the exchange properties are jointly necessary and sufficient for two results under every priority ranking: the greedy rule is the unique choice rule that is non-wasteful, distributionally maximal, and free of justified envy; and it is also path independent. In matching markets, deferred acceptance is the unique mechanism satisfying the three axioms, individual rationality, and strategy-proofness. Our framework accommodates intersectional identities and subsumes models based on reserves and matroids.
%100 words

\medskip
\noindent \textit{Keywords: Distributional preferences, market design, deferred acceptance, path independence.}

\medskip
\noindent \textit{JEL Classifications: D47, C78, I24, I28.}
\end{abstract}

\newpage
\section{Introduction}

Many institutions face multidimensional selection problems. They care about the quality of the individuals they select, but they also care about the composition of the selected group. A university may want to admit strong students while also forming a class that reflects socioeconomic, racial, geographic, or curricular diversity. A hospital may want highly qualified residents while also valuing a cohort with different backgrounds, specialties, and service interests. A firm may want talented employees while also caring about the diversity of skills and experiences in teams. Thus, a central question in such settings is: How should an institution select individuals according to priority or merit while also pursuing distributional objectives over the group as a whole?

In this paper, we answer this question in a canonical school-choice environment. While we focus on school choice, the same design problem arises in labor markets, public-sector hiring, and team formation. We use the school-choice framing because these multidimensional objectives are well documented and prevalent \citep{alves1987}. Schools naturally wish to admit students with the highest priorities, which may be determined by standardized test performance, proximity of residence, and sibling attendance. Yet, they may also wish---and are often legally required---to guarantee access for various identity groups based on gender, race, veteran status, socioeconomic status, and other traits.

Formally, this paper develops a framework that takes as primitive a school's preference over sets of students, which we refer to as the school's \textit{distributional preference}. Such a preference can capture the school's desire to attain a particular composition of identity groups or other relevant traits in the admitted cohort. Importantly, the framework allows us to articulate complex distributional objectives directly, including objectives involving intersectional or overlapping identity groups. Given this preference and a priority ranking over individual students, we ask when the school can choose a cohort that is distributionally desirable, respects priorities, behaves consistently across applicant pools, and can be embedded in a centralized matching mechanism.

This preference-based approach differs from much of the market-design literature on distributional objectives, which often starts from particular admissions policies or choice rules and asks when those policies are consistent with a school's distributional objectives. The focus is often on quotas, reserves, or floors-and-ceilings policies. In general, such policies specify a minimum, maximum, or reserved number of seats to be allocated to students from each identity category. Hence, these policies give a simple language for formalizing the school's distributional objective and often work well when each student belongs to a single identity category. 

When a student belongs to multiple categories, however, these policies may fail to articulate the school's distributional objective over intersectional identities. Counting such a student toward only one of her identities or allowing her to count toward multiple identities both have shortcomings, as the following example illustrates.

\begin{example}\label{example:running}
Consider a school with two seats. It evaluates the composition of its student body based on socioeconomic status (SES) and minority status. There are four applicants, $\mathcal{S}=\{s_1,s_2,s_3,s_4\}$, prioritized in descending order so that $s_1$ has the highest priority and $s_4$ has the lowest. Student $s_1$ is both a minority and a low-SES student, $s_2$ is neither a minority nor a low-SES student, $s_3$ is a low-SES student, and $s_4$ is a minority student.

The school has two reserve categories: one for low-SES students and one for minority students. A reserve category is filled by the highest-priority eligible applicant who has not already been admitted. Once all reserve categories are filled, any remaining seats are assigned according to the priority ranking.

Suppose first that an intersectional student can count toward at most one reserve category. If the SES reserve is processed first, then $s_1$ fills the SES reserve; the minority reserve then goes to $s_4$, yielding $\{s_1,s_4\}$ as the admitted pair. If instead the minority reserve is processed first, then $s_1$ fills the minority reserve; the SES reserve then goes to $s_3$, yielding $\{s_1,s_3\}$. Thus, the outcome depends on the order in which the reserve categories are processed. Yet the canonical description of reserve policies does not determine which reserve category should be processed first.

Next, suppose intersectional students count toward both reserve categories simultaneously. In this case, regardless of which reserve category is processed first, $s_1$ fulfills both reserves. The remaining seat goes to the highest-priority remaining student, $s_2$, yielding the set $\{s_1,s_2\}$.

However, this violates path independence. If $s_1$ is removed from the applicant pool, then, in the reduced pool $\{s_2,s_3,s_4\}$, the highest-priority students who qualify for the SES and minority reserves are $s_3$ and $s_4$, respectively. Thus, the school selects $\{s_3,s_4\}$ and rejects $s_2$, who qualifies for neither reserve category. Hence, $s_2$ is admitted from the larger applicant pool but rejected after another applicant is removed. This inconsistency is normatively undesirable in a single-school choice problem and can undermine implementation in centralized matching markets.
\end{example}

These shortcomings arise not simply because the school values more than one category. Rather, the reserve policy does not articulate how these categories should be jointly evaluated when a student belongs to more than one of them. More broadly, focusing directly on admission policies or choice rules leaves the school's underlying distributional objective implicit. We thus start from an explicit model of the school's preferences over cohorts and study the properties of admission policies induced by those preferences.

To that end, we first consider choice by a single school. Given an applicant pool, the school's distributional preference identifies a \emph{frontier}: the collection of non-wasteful cohorts that are undominated according to the distributional preference. We then define the \emph{distributional choice rule}, which is a greedy algorithm that considers students in priority order and admits a student whenever doing so remains compatible with eventually reaching some set on the frontier. Thus, the distributional preference determines which final cohorts lie on the frontier, while the priority ranking selects among those cohorts.

By construction, the distributional choice rule is non-wasteful, distributionally maximal, and free of justified envy. These three axioms capture the basic requirements that the school fill its available seats whenever possible, select an undominated cohort according to its distributional preference, and depart from the priority ranking only when doing so is justified by that preference. However, without imposing further structure on the distributional preference, these axioms neither uniquely identify the distributional choice rule nor ensure that the rule behaves consistently across applicant pools.

We therefore identify three structural properties of the distributional preference. The \emph{upper-bound property} imposes a ``lattice-like'' structure on the preference by requiring incomparable non-wasteful cohorts to have a distributional upper bound. As a result, every cohort on the frontier is not merely maximal but optimal among the non-wasteful cohorts. The \emph{maximizer property} then requires optimal cohorts within a fixed applicant pool to be connected by value-preserving one-for-one exchanges. Specifically, given two optimal cohorts, there must exist one student belonging to each cohort but not the other whose exchange leaves both cohorts distributionally indifferent to their corresponding original cohorts. Finally, the \emph{coherence property} imposes a corresponding exchange condition across applicant pools when students are added or removed.

These properties deliver our main single-school results. Our first set of results concerns choice from a fixed applicant pool. Given the upper-bound property, the maximizer property is necessary and sufficient for the distributional choice rule to be the unique choice rule that is non-wasteful, distributionally maximal, and free of justified envy. Moreover, any departure from the distributional choice rule must either leave some seats empty, select a distributionally dominated cohort, or select a cohort that is inferior according to the priority ranking. We also provide a revealed-priority foundation for the rule. When the priority ranking is not taken as primitive but rather is inferred from the school's choices, non-wastefulness, distributional maximality, and acyclicity of revealed priority comparisons imply that the choice rule must be the distributional choice rule for some priority ranking.

Our second set of results concerns consistency of choice across applicant pools. Given the upper-bound property, the maximizer and coherence properties are jointly necessary and sufficient for the distributional choice rule to satisfy path independence for every priority ranking. 

These results also clarify the choices in
\autoref{example:running} by making the relevant comparisons over cohorts
explicit. Suppose that $\{s_1,s_3\}$ belongs to the frontier and weakly
dominates both $\{s_1,s_4\}$ and $\{s_3,s_4\}$. Such comparisons are
natural if the school places relatively greater weight on low-SES
representation. The decisive remaining comparison is between
$\{s_1,s_3\}$ and $\{s_1,s_2\}$. If
$\{s_1,s_3\}\succ\{s_1,s_2\}$, then $\{s_1,s_2\}$ does not belong to
the frontier. The distributional choice rule therefore admits $s_1$,
rejects $s_2$, and admits $s_3$, yielding $\{s_1,s_3\}$. If instead
$\{s_1,s_3\}\sim\{s_1,s_2\}$, then both cohorts belong to the frontier.
Because $s_2$ has higher priority than $s_3$, the rule selects
$\{s_1,s_2\}$. In this way, the
preference-based approach replaces an exogenous convention about how to
count an intersectional student with explicit comparisons over the
resulting cohorts.

We next extend the analysis to centralized matching markets. We show that when every school's distributional preference satisfies the three structural properties, deferred acceptance based on the induced distributional choice rules is the unique mechanism that is individually rational, non-wasteful, distributionally maximal, strategy-proof, and free of justified envy. Thus, the same structural properties that ensure a well-behaved choice rule for a single school also support deferred-acceptance implementation in a market with multiple schools.

The framework both rationalizes familiar policies and extends beyond them. Standard and overlapping reserve rules arise from distributional preferences based on partition and transversal matroids, respectively. The framework also includes a graded formulation of floors and ceilings, represented by distributional preferences that rank violations of lower and upper bounds on the representation of each identity group. It further accommodates objectives that are not naturally expressed through reserves, including score-based preferences that assign different values to different combinations of intersectional identities, preferences that value linearly non-redundant combinations of traits, and preferences that value linkages across identity groups. Thus, the paper provides a unified preference foundation for distributional policies whose induced choice rules respect priorities, are non-wasteful, and support deferred-acceptance implementation.

\subsection{Related Literature}

The market-design literature on distributional policies has largely developed through reserves, quotas, and related constraints. Standard reserve rules were introduced and studied by \citet{hayeyi13}, floors-and-ceilings rules by \citet{ehayeyi14}, and related axiomatic foundations were provided by \citet{echyen12}. Reserve policies have been applied in Brazil \citep{aybo16}, Chile \citep{correa19}, and India \citep{aygtur17,sonmez/yenmez:22}. Further work studies regional constraints \citep{kamakoji-basic,kamakoji-concepts,kamakoji-iff}, reserve processing order \citep{dur_boston,dur16}, and alternative systems of quotas, transfers, and reserve adjustment \citep{westkamp10,fratro:2017,aytur2020,kumano/kurino:2022,dogan/erdil:2025}. \citet{imamura2025} provides a recent axiomatic analysis of the trade-off between meritocracy and diversity.

Much of this literature begins from reserve categories, eligibility requirements, distributional targets, or quota constraints. Related preference-based work studies the rationalizability of affirmative-action choice rules and the implications of multidimensional identities \citep{celebi:2024}. Our analysis takes as primitive a possibly incomplete distributional preference over cohorts and keeps it separate from the priority ranking over individual students. This formulation allows us to characterize when the two can be combined through a uniquely characterized and path-independent priority-based choice rule.

The paper most closely related to ours is \citet{hakoyeyo2022}. They represent an institution's distributional objective by a real-valued function of type distributions and study a distribution-first, merit-second choice rule. Under ordinal concavity, they show that their distribution-conscious rule maximizes merit among distributionally optimal cohorts, can be computed in polynomial time, and satisfies path independence; their analysis also develops related matroid characterizations. Our analysis instead imposes non-wastefulness, allows the primitive distributional preference to be incomplete, and interprets the ranking over individuals as a system of priority rights. Conditional on the upper-bound property, we provide necessary and sufficient conditions for the greedy priority refinement to be uniquely characterized by non-wastefulness, distributional maximality, and freedom from justified envy, and to satisfy path independence under every priority ranking. When restricted to fixed-cardinality cohorts, the ordinal concavity condition studied by \citet{hakoyeyo2022} induces distributional preferences satisfying our structural properties.

A complementary rationalization result is provided by \citet{yokoteetal2025}, who show that a choice rule is path independent if and only if it can be rationalized by an ordinally concave utility function. Their result asks whether a given choice rule admits an appropriate utility representation. We instead hold a primitive distributional preference fixed and characterize when combining it with any priority ranking through the greedy rule yields path-independent choice. More broadly, \citet{kojima-tamura-yokoo} use M$^{\natural}$-concavity to design matching mechanisms under distributional constraints. \citet{majorization2024} take a different but related approach, using a generalization of majorization to compare diversity relative to distributional targets and to axiomatize the associated choice rule.

Our treatment of overlapping identities is closely related to the literature on horizontal and overlapping reserves \citep{sonmez/yenmez:20,sonmez/yenmez:22}. These policies assign each candidate to at most one eligible reserve category while filling as many reserved positions as possible. We show that the resulting rule can be rationalized by a distributional preference based on the rank function of a transversal matroid; standard reserves arise analogously from partition matroids. More broadly, \citet{flein03} shows how greedy choice under matroidal constraints can be combined with deferred acceptance, while \citet{bonet2024} introduce \emph{explainability} for affirmative-action policies and characterize outcome-based greedy rules. Whereas these approaches begin with a reserve system, feasible-outcome family, or matroidal constraint, we characterize when such a structure is generated by a primitive distributional preference. Conditional on the upper-bound property, the maximizer property holds if and only if each frontier forms the bases of a matroid, while the coherence property identifies when these frontiers are consistent across applicant pools.

We also contribute to the characterization of deferred-acceptance mechanisms. \citet{albar94} characterize deferred acceptance when priorities are exogenously given; alternative characterizations appear in \citet{kojman10} and \citet{ehlkla16}. \citet{abdulkadiroglu_grigoryan2025} study deferred acceptance coupled with a reserves-and-quotas choice rule that minimizes priority violations subject to distributional targets and other axioms. We instead begin from primitive distributional preferences and use the meta-characterization of \citet{Dogan/Imamura/Yenmez:2025} to lift our choice-rule characterization to the mechanism level. This yields the unique deferred-acceptance mechanism satisfying the corresponding requirements of non-wastefulness, distributional maximality, and freedom from justified envy together with individual rationality and strategy-proofness. We also provide a revealed-priority foundation for the greedy distributional choice rule when the priority ranking is not taken as primitive.

Finally, the two-stage structure of the distributional choice rule is related to sequentially rationalizable choice. \citet{manzini2007sequentially} study procedures in which ordered rationales successively refine the set of alternatives. In our setting, the distributional preference first determines the frontier of distributionally maximal non-wasteful cohorts, and the priority ranking then selects among those cohorts. The two criteria represent an institutional objective and individual priority rights, respectively. Our results characterize when their ordered application yields a unique path-independent choice rule and supports implementation through deferred acceptance.

The rest of the paper is organized as follows: \autoref{sec:model} introduces the model. \autoref{sec:choice} develops the structural properties and establishes the single-institution results. \autoref{sec:two-sided} extends the analysis to centralized markets. \autoref{sec:applications} presents applications, and \autoref{sec:conclusion} concludes. All proofs are in the Appendix.

\section{Model}\label{sec:model}

\subsection{Preliminaries}
We denote the set of non-negative integers by $\mathbb{Z}_{+}$ and the set of positive integers by $\mathbb{Z}_{++}$. For any finite set $\mathcal{E}$, let $2^{\mathcal{E}}$ denote its power set and $\abs{\mathcal{E}}$ its cardinality. A \textit{strict linear order} on $\mathcal{E}$ is a transitive, asymmetric, and total binary relation.

A \textit{preorder} $\succsim$ on $2^{\mathcal{E}}$ is a reflexive and transitive binary relation. For $E, E' \subseteq \mathcal{E}$, we write $E \succ E'$ if $E \succsim E'$ and $E' \not\succsim E$, and $E \sim E'$ if $E \succsim E'$ and $E' \succsim E$. Sets $E$ and $E'$ are \textit{incomparable} under $\succsim$ if neither $E \succsim E'$ nor $E' \succsim E$.

A \emph{matroid} is a pair $(\mathcal{E},\, \mathcal{I})$, where $\mathcal{E}$ is a finite \textit{ground set} and $\mathcal{I}$ is a collection of subsets of $\mathcal{E}$, called \emph{independent sets}, satisfying
\begin{itemize}[parsep = 0pt, listparindent=1em, itemsep=0pt]
    \item[(I1)] \textit{Non-emptiness}: $\varnothing\in \mathcal{I}$. 
    \item[(I2)] \textit{Hereditary}: If $I \in \mathcal{I}$ and $I' \subseteq I$, then $I' \in \mathcal{I}$.
    \item[(I3)] \textit{Augmentation}: If $I, I' \in \mathcal{I}$ and $\abs{I} > \abs{I'}$, then there exists $e \in I \setminus I'$ such that $I' \cup \{e\} \in \mathcal{I}$.
\end{itemize}
An independent set that is maximal under set inclusion is called a \textit{basis}, i.e., $B\in \mathcal{I}$ is a basis if there exists no $I\in\mathcal{I}$ such that $I\supsetneq B$. 

\subsection{Setup}\label{subsec:setup}

We begin with a single school and a finite, nonempty set of students $\mathcal{S}$. The school has capacity $q \in \mathbb{Z}_{++}$, a \emph{priority ranking} $\pi$ over students, and a \emph{distributional preference} $\succsim$ over subsets of students. The priority ranking is a strict linear order on $\mathcal{S}$ and represents the school's ordering over individual students, often based on a student's test performance, proximity of residence, or sibling attendance. The distributional preference is a preorder on $2^{\mathcal{S}}$ and represents the school's preferences over the composition of student groups. The tuple $(\mathcal{S},\pi,\succsim,q)$ thus constitutes the primitives of the single-school model. We extend this model to two-sided matching markets with multiple schools in \autoref{sec:two-sided}.

We take the priority ranking $\pi$ and the distributional preference $\succsim$ as distinct from one another. In particular, we do not assume that $\succsim$ is derived from or related to $\pi$ in any systematic way. One useful interpretation is that each student has a collection of demographic traits, and that $\succsim$ compares subsets of students according to their aggregate trait composition. Under such a preference, two sets $S,S'\subseteq \mathcal{S}$ may satisfy $S\sim S'$ whenever they have the same trait composition, even if they contain different students. 

Because we do not require $\succsim$ to be complete (recall that it is a preorder), it need not compare all pairs of sets. For instance, a set consisting only of low-income students and a set consisting only of minority students may be incomparable if the school does not rank one identity category against the other. At the same time, $\succsim$ can compare sets involving overlapping identities, as in the motivating example in the Introduction. We illustrate the generality of the framework using the following examples.

\begin{example}\label{example:2}
Consider the four students in \autoref{example:running}. Suppose each student $s$ receives a distributional score $f(s)$ given by
\[
f(s)=\begin{cases}
    2 & \text{if } s \text{ is both minority and low-SES},\\
    1 & \text{if } s \text{ is minority but not low-SES},\\
    0.5 & \text{if } s \text{ is low-SES but not minority},\\
    0 & \text{otherwise}.
\end{cases}
\]
Thus, the school assigns different scores to different combinations of traits. Among students with exactly one protected trait, the school values minority students more than it does students with low-SES. At the same time, it also places a premium on intersectionality with $f(s_3)+f(s_4)<f(s_1)$, so that a student with both traits has a higher score than the combined value of the two traits when held by separate students.

We derive a distributional preference over $2^\mathcal{S}$ from the student-level distributional scores as follows: for any $S,S'\subseteq\mathcal{S}$, let $S\succsim S'$ if and only if $\sum_{s\in S}f(s)\geq \sum_{s\in S'}f(s)$. The induced distributional preference is complete, so any two subsets are comparable. For example, $ \{s_1,s_2\}\sim \{s_1\}\succ\{s_3,s_4\}$.
\end{example}

\begin{example}\label{example:3} 
Once again, consider the four students in \autoref{example:running}. Let $P_{SE}=\{s_1,s_3\}$ be the set of students eligible for a protected category based on socioeconomic status, and let $P_M=\{s_1,s_4\}$ be the set of students eligible for a protected category based on minority status. For any $S\subseteq\mathcal{S}$, let $r(S)$ denote the maximum number of protected categories that can be represented by assigning students in $S$ to categories for which they are eligible, with each student assigned to at most one protected category.

We then derive a distributional preference over $2^\mathcal{S}$ from $r(\cdot)$: for any $S,S'\subseteq\mathcal{S}$, let $S\succsim S'$ if and only if $\abs{S}=\abs{S'}$ and $r(S)\geq r(S')$. Such a preference reflects a school that cares about how many protected categories are represented in its student body, perhaps because it receives additional funding for each represented category, but does not compare student bodies of unequal size. In this case, $\{s_3,s_4\}\succ\{s_1,s_2\}$, but neither set is comparable to $\{s_1\}$.
\end{example}

\begin{example}\label{example:vector}
Consider the four students in \autoref{example:running}. Using minority status as the first coordinate and low-SES status as the second, define each student's type vector as
\[
\tau(s_1)=(1,1),\qquad \tau(s_2)=(0,0),\qquad \tau(s_3)=(0,1),\qquad \tau(s_4)=(1,0).
\]
Suppose the school evaluates cohorts based on the span of the cohort's type vectors, i.e., the school's preference captures a desire for cohorts with non-redundant trait profiles.

In this case, the pair $\{s_1,s_2\}$ has a span of dimension one: $s_1$ contributes the type vector $(1,1)$, while $s_2$ contributes the zero vector. By contrast, $\{s_3,s_4\}$ has a span of dimension two because the vectors $(0,1)$ and $(1,0)$ are linearly independent over $\mathbb{R}$. Hence, $\{s_3,s_4\}\succ\{s_1,s_2\}$.
\end{example}

\begin{example}\label{example:phd}
Consider a PhD program in computational social science that offers four fields: Computer Science (CS), Machine Learning (ML), Political Science (PS), and Economics (Econ). It wants to admit a cohort that connects as many of these four fields as possible. 

There are four candidates: $s_1$ with interests in CS and Econ (algorithmic game theory); $s_2$ with interests in ML and PS (natural language processing for political text); $s_3$ with interests in CS and ML (algorithmic foundations of machine learning); and $s_4$ with interests in Econ and ML (econometric machine learning). 

The cohort $\{s_1,s_2,s_3\}$ connects all four fields, while $\{s_1,s_3,s_4\}$ connects only three of them, as depicted by the graphs in \autoref{fig:graphic-matroid-example}. Hence, the PhD program's distributional preference satisfies $\{s_1,s_2,s_3\}\succ\{s_1,s_3,s_4\}$. Importantly, every student in this example is intersectional. Yet the program does not evaluate students only by the number or identity of their fields of interest; it also values the pattern of linkages represented by the cohort.

\begin{figure}[htbp]
\centering

\begin{subfigure}[t]{0.47\textwidth}
\centering
\begin{tikzpicture}[
    every node/.style={circle, draw, minimum size=8mm, inner sep=1pt},
    lab/.style={draw=none, fill=white, inner sep=1pt, font=\small}
]
    % vertices
    \node (CS) at (0,2) {CS};
    \node (ML) at (2.5,2) {ML};
    \node (Econ) at (0,0) {Econ};
    \node (PS) at (2.5,0) {PS};

    % edges
    \draw (CS) -- node[lab, midway, left] {$s_1$} (Econ);
    \draw (CS) -- node[lab, midway, above] {$s_3$} (ML);
    \draw (ML) -- node[lab, midway, right] {$s_2$} (PS);
\end{tikzpicture}
\caption{Cohort \(\{s_1,s_2,s_3\}\)}
\end{subfigure}
\hfill
\begin{subfigure}[t]{0.47\textwidth}
\centering
\begin{tikzpicture}[
    every node/.style={circle, draw, minimum size=8mm, inner sep=1pt},
    lab/.style={draw=none, fill=white, inner sep=1pt, font=\small}
]
    % vertices
    \node (CS) at (0,2) {CS};
    \node (ML) at (2.5,2) {ML};
    \node (Econ) at (0,0) {Econ};
    \node (PS) at (2.5,0) {PS};

    % edges
    \draw (CS) -- node[lab, midway, left] {$s_1$} (Econ);
    \draw (CS) -- node[lab, midway, above] {$s_3$} (ML);
    \draw (Econ) -- node[lab, midway, below right] {$s_4$} (ML);
\end{tikzpicture}
\caption{Cohort \(\{s_1,s_3,s_4\}\)}
\end{subfigure}

\caption{Each admitted student contributes an undirected edge linking the two fields they bridge. The cohort \(\{s_1,s_2,s_3\}\) connects all four fields, whereas \(\{s_1,s_3,s_4\}\) connects only three, leaving Political Science isolated.}
\label{fig:graphic-matroid-example}
\end{figure}
\end{example}\medskip

Although we take $\succsim$ as a primitive preorder on $2^{\mathcal{S}}$, this does not require a policymaker to list, or a clearinghouse to compute, preferences over all subsets. As shown in \autoref{example:2}, distributional preferences may be derived from student-level distributional scores. More broadly, we show in \autoref{sec:applications} that distributional preferences often admit compact representations through diversity indices, penalty functions, or matroid rank functions. These representations evaluate sets of students through summary statistics of their trait composition, making the preference interpretable and operational.

\subsection{Frontier Sets and Choice-Rule Desiderata}

A school can only admit a subset of the students that apply to it, and cannot do so in a way that violates its capacity constraint. Thus, the school's distributional goals are constrained by both its applicant pool and its capacity. To formalize these constraints, we introduce two concepts that will be useful in our analysis. 

For any set of students $S \subseteq \mathcal{S}$, let
\[
\nw(S) \coloneqq \left\{S'\subseteq S : \abs{S'}=\min\{q,\abs{S}\} \right\}
\]
denote the collection of \emph{non-wasteful} subsets of $S$: the subsets of $S$ that contain as many students as possible, up to the school's capacity. Let
\[
\front(S)\coloneqq \left\{S' \in \nw(S) : \nexists S'' \in \nw(S) \text{ such that } S'' \succ S' \right\}
\]
denote the \emph{frontier} of $S$. The frontier consists of all non-wasteful subsets that are undominated, or maximal, with respect to the distributional preference. For any set $S\subseteq \mathcal{S}$, $\nw(S)$ is a finite, non-empty collection of subsets of students. Since every finite nonempty preorder has at least one maximal element, $\front(S)$ is nonempty.\footnote{Note that even when $S=\varnothing$,  $\nw(S)=\front(S)=\{\varnothing\}$ are nonempty sets.}

We now turn to our main question: \textbf{How should a school allocate its limited number of seats in a way that balances its distributional goals and priority considerations?} We address this question by imposing a set of normative choice desiderata, formalized as axioms, on how a school chooses which students to admit.

Formally, the school's \emph{choice rule} is a function $\ch:2^\mathcal{S}\to 2^\mathcal{S}$ that maps each set of applicants $S\subseteq \mathcal{S}$ into a set of admitted students $\ch(S)$ such that
\begin{enumerate}[label={$(\alph*)$}, parsep=0pt, itemsep=2pt]
    \item $\ch(S)\subseteq S$, and
    \item $\abs{\ch(S)}\le q$.
\end{enumerate}

Our first axiom imposes a weak efficiency constraint: no seat may remain empty while an eligible applicant is available.

\begin{axiom}\label{ax:nonwasteful}
A choice rule $\ch$ is \textbf{non-wasteful} if, for every set $S\subseteq \mathcal{S}$,
\[
\abs{\ch(S)}=\min\{q,\abs{S}\}.
\]
\end{axiom}

Non-wastefulness precludes choice rules that may purposefully leave seats vacant, such as quota-based rules that admit students of a given trait only up to a specified cap. Nevertheless, the axiom is natural in our primary applications, such as public school choice and civil service hiring, which operate under statutory or funding constraints that require seats to be filled whenever acceptable applicants are available.\footnote{Unacceptable applicants can be removed from $\mathcal{S}$ without affecting the analysis.}

For each set $S$, non-wastefulness requires $\ch(S)$ to belong to $\nw(S)$. However, when the applicant size $\abs{S}$ exceeds the school's capacity $q$, $\nw(S)$ is not a singleton. Therefore, non-wastefulness alone does not generally determine which students should be admitted. In such cases, our second axiom requires the school to select an undominated, or maximal, set with respect to its distributional objective while holding fixed the number of admitted students.

\begin{axiom}\label{ax:diversity}
A choice rule $\ch$ is \textbf{distributionally maximal} if, for every set $S\subseteq\mathcal{S}$ and $S'\subseteq S$,
\[
\abs{S'}=\abs{\ch(S)} \implies S'\not\succ \ch(S).
\]
\end{axiom}

The first two axioms are jointly equivalent to requiring the chosen set to be in the frontier.

\begin{lemma}\label{lem:distfrontier}
A choice rule $\ch$ is non-wasteful and distributionally maximal
if and only if $\ch(S)\in\front(S)$ for every set $S\subseteq\mathcal{S}$.
\end{lemma}

The lemma establishes that the frontier is the relevant benchmark for evaluating non-wasteful and distributionally maximal choice rules. However, the school's distributional preference $\succsim$ may leave many subsets indifferent or incomparable, so $\front(S)$ may contain multiple sets. Our third axiom requires the priority ranking to guide the school's choice among such sets. It imposes a fairness requirement: the school cannot reject a higher-priority student in favor of a lower-priority admitted student whenever admitting the higher-priority student instead would weakly improve the distributional outcome.

\begin{axiom}\label{ax:envy}
A choice rule $\ch$ is \textbf{free of justified envy} if, for every $S\subseteq\mathcal{S}$, $s\in\ch(S)$, and $s'\in S\setminus\ch(S)$,
\[
\bigl(\ch(S)\setminus\{s\}\bigr)\cup\{s'\}\succsim \ch(S)
\implies s\mathrel{\priority}s'.
\]
\end{axiom}

The three axioms above are the central desiderata for the single-school choice problem: the school must fill its seats guided by its distributional objective while also respecting priorities whenever doing so is compatible with the distributional objective. Our single-school characterization in \autoref{sec:choice} focuses primarily on these three requirements. 

Our fourth, and final, axiom plays a limited role in single-school choice problems, but it is central to our market-design analysis in \autoref{sec:two-sided}. The axiom imposes a consistency condition: processing the applicant pool all at once or in batches should not change the admitted set.

\begin{axiom}
A choice rule $\ch$ is \textbf{path independent} if, for every $S,S'\subseteq\mathcal{S}$,
\[
\ch(S\cup S')=\ch\left(\ch(S)\cup S'\right).
\]
\end{axiom}

Path independence is normatively appealing in single-school choice problems \citep{echyen12}, and it underlies existence and structural results in matching markets \citep{blair88,alkan03,flein03}, making it a central axiom in two-sided market design problems.

\subsection{Distributional Choice Rule}\label{subsec:distributional}

We now introduce the greedy choice rule of interest, which we refer to as the \emph{distributional choice rule} and denote by $\ch^\priority$.

\noindent\textbf{Distributional Choice Rule:}
\begin{quoting}[leftmargin=.5cm, rightmargin=0cm]
\begin{description}[parsep=0pt, itemsep=0pt, leftmargin=0em]
    \item[Input:] A set of students $S\subseteq\mathcal{S}$. Let $k\coloneqq \abs{S}$. Label students in $S$ such that $s_1\mathrel{\priority}s_2\mathrel{\priority}\cdots \mathrel{\priority}s_k$.
        \item[Initialization:] Set $S_0\coloneqq\varnothing$.
        \item[Step $i\in\{1,\dots,k\}$:] If $S_{i-1}\cup\{s_i\}\subseteq S'$ for some $S'\in\front(S)$, set $S_i\coloneqq S_{i-1}\cup\{s_i\}$. Otherwise, set
        $S_i\coloneqq S_{i-1}$.
        \item[Output:] $\ch^{\priority}(S)\coloneqq S_k$.
    \end{description}
\end{quoting}

Let us illustrate how the distributional choice rule combines priority and distributional preferences using \autoref{example:2} and \autoref{example:3}. Suppose the school has capacity $q=2$ and the priority ranking is $s_1\mathrel{\priority}s_2 \mathrel{\priority} s_3\mathrel{\priority}s_4$.

Consider first the score-based preference in \autoref{example:2}. Among non-wasteful sets, the maximal total score is $3$, with
$\front(\mathcal S)=\big\{\{s_1,s_4\}\big\}$.
The algorithm first considers $s_1$. Since $s_1$ belongs to the frontier set $\{s_1,s_4\}$, the rule admits $s_1$. It next sequentially considers $s_2$ and $s_3$, and rejects both because neither $\{s_1,s_2\}$ nor $\{s_1,s_3\}$ belongs to the frontier. Finally, the rule considers $s_4$ and admits her since $\{s_1,s_4\}$ is in the frontier. Hence, $\ch^\priority(\mathcal S)=\{s_1,s_4\}$. 

Next, consider the distributional preference in \autoref{example:3}. In this case, the school values the number of protected categories that can be represented when assigning each admitted student to at most one category. The frontier consists of all non-wasteful sets that represent both protected categories. Therefore,
\[
\front(\mathcal S)=\big\{\{s_1,s_3\},\{s_1,s_4\},\{s_3,s_4\}\big\}.
\]
The algorithm first considers and admits $s_1$, since $s_1$ is contained in some frontier set. It then considers and rejects $s_2$, since no frontier set contains $s_2$. Next, it considers and admits $s_3$, since the set $\{s_1,s_3\}$ is in the frontier. Once these two students have been admitted, no additional student can be added while remaining contained in a two-student frontier set. Thus, $\ch^\priority(\mathcal S)=\{s_1,s_3\}$. 

The two examples have the same students, capacity, and priority ranking. Yet, each distributional preference systematically induces a different frontier and therefore a different choice. More broadly, at each step of the algorithm, $\ch^\priority$ greedily admits the highest-priority student whose admission remains compatible with eventually reaching some set on the frontier. In fact, by construction, the output of the algorithm lies in $\front(S)$ for each applicant set $S$. As a result, by \autoref{lem:distfrontier}, $\ch^\priority$ is non-wasteful and distributionally maximal.

The rule is also free of justified envy. To see why, consider two students $s$ and $s'$ with $s\in\ch^\priority(S)$ and $s'\in S\setminus\ch^\priority(S)$. Suppose, toward a contradiction, that $s'\mathrel{\priority}s$ and
\[
\bigl(\ch^\priority(S)\setminus\{s\}\bigr)\cup\{s'\}\succsim \ch^\priority(S).
\]
Since $\ch^\priority(S)\in\front(S)$, the second condition implies that $\bigl(\ch^\priority(S)\setminus\{s\}\bigr)\cup\{s'\}\in \front(S)$. Indeed, $\bigl(\ch^\priority(S)\setminus\{s\}\bigr)\cup\{s'\}$ is non-wasteful and weakly preferred to a set in the frontier, so it cannot be strictly dominated by any non-wasteful subset of $S$. Moreover, because $s'\mathrel{\priority}s$, student $s'$ is considered before $s$ in the algorithm. Let $i$ be the step at which $s'$ is considered. Since $s'$ is rejected, $S_i=S_{i-1}$. Moreover, because $s$ is considered later, $S_{i-1}\subseteq \ch^\priority(S)\setminus\{s\}$. Hence, $S_{i-1}\cup\{s'\}\subseteq \bigl(\ch^\priority(S)\setminus\{s\}\bigr)\cup\{s'\}\in\front(S)$. However, this implies that the algorithm would have admitted $s'$ at step $i$, leading to a contradiction. Therefore $\ch^\priority$ is free of justified envy.

Thus, $\ch^\priority$ is non-wasteful, distributionally maximal, and free of justified envy. However, showing that $\ch^\priority$ is the unique choice rule satisfying these three axioms or that it satisfies path independence requires imposing additional structure on the distributional preference. We identify that structure in the following section.

\section{Choice with Distributional Preferences}\label{sec:choice}

In this section, we identify two structural properties of distributional preferences under which our axioms uniquely characterize the distributional choice rule for a fixed applicant pool. We then introduce a third property that, together with the first two, extends this characterization across applicant pools and ensures path independence. 

The structural properties introduced below restrict $\succsim$ only over sets of cardinality $q$. This is without loss for our analysis. If $\abs{S}\leq q$, non-wastefulness requires $\ch(S)=S$, so neither the distributional preference nor the priority ranking affects the choice. If $\abs{S}>q$, every non-wasteful subset of $S$ has cardinality $q$, so all relevant distributional comparisons are among sets of cardinality $q$. Consequently, when specifying applications in \autoref{sec:applications}, it is sufficient to define the distributional objective over sets of size $q$.

The first property---the \textit{upper-bound property}---ensures that incomparability among sets of cardinality $q$ does not create a decision deadlock. If two such sets are incomparable, then their union must contain a non-wasteful subset that is strictly preferred to at least one of them. 

\begin{definition}\label{def:upperbound}
The distributional preference $\succsim$ satisfies the \textbf{upper-bound property} if, whenever two sets $S,S'\subseteq\mathcal{S}$ are incomparable under $\succsim$ and satisfy $\abs{S}=\abs{S'}=q$, there exists $S''\in\nw(S\cup S')$ such that $S''\succ S$ or $S''\succ S'$.
\end{definition}

By definition, the frontier consists of non-wasteful sets that are maximal: no other non-wasteful set strictly dominates them. Maximality alone, however, does not imply that a given frontier set is comparable to every other non-wasteful set. As formalized by the following proposition, the upper-bound property rules out such incomparability and guarantees that the frontier coincides with the collection of non-wasteful sets that are not merely maximal but \emph{optimal}, i.e., most-preferred with respect to the distributional preference among all non-wasteful sets.

\begin{proposition}\label{prop:comp}
The distributional preference satisfies the upper-bound property if and only if, for each set $S\subseteq \mathcal{S}$, the following hold:
\begin{enumerate}[label={$(\roman*)$}, parsep=0pt, itemsep=2pt]
    \item For every $S',S''\in\front(S)$, we have $S'\sim S''$.
    \item For every $S'\in\front(S)$ and every $S''\in\nw(S)\setminus\front(S)$, we have $S'\succ S''$.
\end{enumerate}
\end{proposition}

The second property---the \textit{maximizer property}---imposes a symmetric exchange condition on optimal non-wasteful sets.

\begin{definition}\label{def:maximizer}
The distributional preference $\succsim$ satisfies the \textbf{maximizer property} if, for every distinct $S,S'\subseteq\mathcal{S}$ with $\abs{S}=\abs{S'}=q$, whenever $S\succsim S''$ and $S'\succsim S''$ hold for every $S''\in\nw(S\cup S')$, there exist $s\in S\setminus S'$ and $s'\in S'\setminus S$ such that
\[
(S\setminus\{s\})\cup\{s'\}\sim S
\quad\text{and}\quad
(S'\setminus\{s'\})\cup\{s\}\sim S'.
\]
\end{definition}
The maximizer property says that optimal non-wasteful sets are not isolated from one another. If two distinct sets $S$ and $S'$ of cardinality $q$ are optimal among the non-wasteful subsets of their union, then there exists a student in $S\setminus S'$ and a student in $S'\setminus S$ who can be exchanged so that each exchanged set is distributionally indifferent to the corresponding original set. Such an exchange moves each set one step toward the other by increasing their overlap. Under the upper-bound property and \autoref{prop:comp}, the maximizer property implies that if $S,S'\in \front(S\cup S')$, then the exchanged sets remain in $\front(S\cup S')$. In this sense, the maximizer property is a discrete connectedness condition on the frontier, analogous to the connectedness of the set of optima in continuous convex optimization problems.

\begin{proposition}\label{lem:matroidbase}
Suppose that the distributional preference satisfies the upper-bound property. Then the maximizer property is satisfied if and only if, for every set $S\subseteq \mathcal{S}$, $\front(S)$ forms the bases of a matroid on the ground set $S$.
\end{proposition}

Together with the upper-bound property, the maximizer property ensures that the frontier forms the bases of a matroid, which in turn supports the greedy characterization of the distributional choice rule. To see that this structure does not hold without the maximizer property, consider the following example.

\begin{example}
Let $\mathcal{S}=\{s_1,s_2,s_3,s_4\}$ and $q=2$. Consider a distributional preference where $\{s_1,s_4\}\sim \{s_2,s_3\}\succ T$ for every other $T\in\nw(\mathcal{S})$. This distributional preference satisfies the upper-bound property but the maximizer property fails: $\{s_1,s_4\}$ and $\{s_2,s_3\}$ are both undominated within $\nw(\mathcal{S})$, yet since
$\{s_1,s_4\}\cap \{s_2,s_3\}=\varnothing$,
any swap of a student between them yields either $\{s_1,s_2\}$, $\{s_1,s_3\}$, $\{s_2,s_4\}$, or $\{s_3,s_4\}$, all of which are strictly dominated. Consequently, $\front(\mathcal{S})=\{\{s_1,s_4\},\{s_2,s_3\}\}$ does not form the bases of a matroid.\footnote{Specifically, the base exchange axiom fails; see \ref{sec:matroid_theory}.}
\end{example}\medskip

The maximizer property is therefore necessary and sufficient for the matroidal structure of the frontier and, as the following theorem shows, for the distributional choice rule to be the unique choice rule that is non-wasteful, distributionally maximal, and free of justified envy.

\begin{theorem}\label{thm:exogenouspriority}
Suppose that the distributional preference satisfies the upper-bound property. Then the following statements are equivalent:
\begin{enumerate}[label={$(\roman*)$}, parsep=0pt, itemsep=2pt]
    \item The distributional preference satisfies the maximizer property. 
    \item For every priority ranking $\priority$, $\ch^{\priority}$ is the unique choice rule that is non-wasteful, distributionally maximal, and free of justified envy.
\end{enumerate}
\end{theorem}

Hence, while the distributional choice rule satisfies our first three axioms without any additional assumptions on the distributional preference, it is the unique such choice rule under the upper-bound and maximizer properties. In other words, any rule that differs from $\ch^\priority$ must violate at least one of the three axioms. Under the same two properties, our next result strengthens this uniqueness characterization to a dominance result: At any applicant set where a rule differs from $\ch^\priority$, it must either leave seats unfilled, select a strictly worse distributional outcome, or choose a set of students that are inferior according to the priority ranking. To formalize this last notion, we extend the priority ranking over individual students to a partial order over sets of students.

\begin{definition}
Let $S,S'\subseteq \mathcal S$. We say that $S$ \textbf{priority dominates} $S'$ if $\abs{S} \geq \abs{S'}$ and, for every $k\leq \abs{S'}$, the $k$-th highest-priority student in $S$ is either identical to, or has higher priority than, the $k$-th highest-priority student in $S'$.
\end{definition}
Priority domination extends the school’s priority ranking from individual students to sets of students. A set $S$ priority dominates $S'$ if it contains at least as many students as $S'$ and if, rank by rank, it contains weakly higher-priority students.

\begin{theorem}\label{thm:meritgeneral}
Suppose that the distributional preference satisfies the upper-bound property and the maximizer property. Let $\ch$ be a choice rule such that $\ch(S)\neq\ch^{\priority}(S)$ for some $S\subseteq\mathcal{S}$. Then at least one of the following holds:
\begin{enumerate}[label={$(\roman*)$}, parsep=0pt, itemsep=2pt]
   \item \textbf{Wastefulness:} $\abs{\ch^{\priority}(S)}>\abs{\ch(S)}$.
    \item \textbf{Distributional Inferiority:} $\ch^{\priority}(S)\succ \ch(S)$.
    \item \textbf{Priority Inferiority:} $\ch^{\priority}(S)$ priority dominates $\ch(S)$.
    \end{enumerate}
\end{theorem}

The results so far concern choice from a fixed applicant pool. The upper-bound and maximizer properties ensure that, within each applicant pool, the frontier has the exchange structure needed for the first three axioms to uniquely select the greedy distributional choice rule. Path independence---our last choice axiom---requires more: processing an applicant pool all at once or in batches must lead to the same admitted set. Equivalently, the choice rule must be consistent across varying applicant pools. This consistency requirement is not only desirable for a single-school choice problem, it is also foundational for deferred-acceptance implementation in centralized matching markets. 

For a choice rule to be path independent, the frontiers associated with different applicant pools must also be ``coherent,'' something that the upper-bound and maximizer properties do not ensure by themselves. The next property imposes exactly this cross-pool coherence.

\begin{definition}\label{def:exchange-coherence}
The distributional preference $\succsim$ satisfies the \textbf{coherence property} if, for every distinct $S,S'\subseteq\mathcal{S}$ with $\abs{S}=\abs{S'}=q$ and every $\hat{s}\in S\setminus S'$ satisfying $(S\setminus S')\setminus\{\hat{s}\}\neq\varnothing$, whenever
\begin{enumerate}[label={$(\roman*)$}, parsep=0pt, itemsep=2pt]
    \item  $S\succ S''$ for every $S''\in\nw(S\cup S')$ with $\hat s\notin S''$,
    \item $S\succsim S''$ for every $S''\in\nw(S\cup S')$ with $\hat s\in S''$, and
    \item $S'\succsim S''$ for every $S''\in\nw\big((S\cup S')\setminus\{\hat{s}\}\big)$,
\end{enumerate}
then there exist $s\in (S\setminus S')\setminus\{\hat{s}\}$ and
$s'\in S'\setminus S$ such that
\[
(S\setminus\{s\})\cup\{s'\}\sim S
\quad\text{and}\quad
(S'\setminus\{s'\})\cup\{s\}\sim S'.
\]
\end{definition}
The economic content of coherence can be understood as a comparative-statics property of the frontier. Let $T\coloneqq S\cup S'$ and suppose $\hat{s}$ satisfies the antecedent of the coherence property. Student $\hat s$ is an \textbf{essential} student in $T$: given an applicant set $T$, the school cannot attain a distributionally maximal set without admitting student $\hat s$. Equivalently, every set in $\front(T)$ contains $\hat s$. Once $\hat{s}$ is removed, every set in the frontier of $T$ becomes infeasible. Thus, the frontier of the reduced pool $T\setminus\{\hat{s}\}$ must be disjoint from the frontier of $T$.

Coherence requires this shift in the frontier to be ``continuous'' in the exchange structure of sets. In particular, suppose $S$ is optimal among the non-wasteful sets of $T$, and $S'$ is optimal among the non-wasteful sets of $T\setminus\{\hat{s}\}$. Then coherence requires some student in $(S\setminus S')\setminus\{\hat{s}\}$ to be exchangeable with some student in $S'\setminus S$ while leaving both resulting sets distributionally indifferent to their corresponding original sets. Thus, while removing an essential student shifts the frontier, it cannot make the new frontier isolated (in the exchange sense) from the original one.

\begin{theorem}\label{thm:genpi}
Suppose that the distributional preference satisfies the upper-bound property. Then the following statements are equivalent:
\begin{enumerate}[label={$(\roman*)$}, parsep=0pt, itemsep=2pt]
    \item The distributional preference satisfies the maximizer and coherence properties.
    \item For every priority ranking $\priority$, $\ch^\priority$ satisfies path independence.
\end{enumerate}
\end{theorem}

Thus, given the upper-bound property, the maximizer and coherence properties are both necessary and sufficient for path independence of the distributional choice rule. The maximizer property gives each frontier the matroid structure needed for greedy choice within a fixed applicant pool. Coherence imposes the additional frontier consistency condition needed across different applicant pools. If either property fails, one can construct a priority ranking for which the greedy distributional rule selects a student from the larger pool but rejects that same student after another applicant is removed, violating path independence. For instance, the failure of path independence in \autoref{example:running} can be attributed to a failure of the maximizer property.\footnote{In \autoref{example:running}, both $\{s_1, s_2\}$ and $\{s_3,s_4\}$ are in the frontier. However, for any $s\in\{s_1,s_2\}$ and $s'\in\{s_3,s_4\}$, at least one of the two exchanged sets must be either $\{s_2, s_3\}$ or $\{s_2, s_4\}$. Both sets represent only one protected category and are therefore strictly dominated by the frontier sets $\{s_1,s_2\}$ and $\{s_3,s_4\}$. As a result, the maximizer property does not hold. }

\subsection{Revealed Priorities}\label{subsec:revealed_priorities}

We now consider a setting in which the priority ranking is not taken as primitive. Instead, priority comparisons are inferred from the school's choice behavior. Suppose a selected student $s$ is chosen over a rejected student $s'$ even though replacing $s$ with $s'$ would weakly improve the distributional outcome. Such a choice reveals $s$ to have higher priority than $s'$. The next axiom requires these revealed priority comparisons to be acyclic.

\begin{axiom}
A choice rule $\ch$ satisfies the \textbf{revealed priority axiom} if there are no sequences of applicant sets $\{S_k\}_{k=1}^K$ and students $\{s_k\}_{k=1}^K$ such that, for every $k\in\{1,\ldots,K\}$,
\begin{enumerate}[label={$(\roman*)$}, parsep=0pt, itemsep=2pt]
    \item $s_k\in \ch(S_k)$,
    \item $s_{k+1}\in S_k\setminus \ch(S_k)$, and 
    \item $\big(\ch(S_k)\setminus\{s_k\}\big)\cup\{s_{k+1}\}\succsim \ch(S_k)$,
\end{enumerate}
where $s_{K+1}=s_1$.
\end{axiom}

The condition $s_{K+1}=s_1$ closes the sequence into a cycle. Thus, the axiom rules out choice behavior that reveals $s_1$ above $s_2$, $s_2$ above $s_3$, and so on, while also revealing $s_K$ above $s_1$. In this sense, rather than starting from an exogenously given priority ranking, we infer one from the school's choice behavior. The axiom then requires these inferred priority comparisons to be acyclic, which is precisely what is needed for them to be extended to a strict linear priority ranking.

\begin{theorem}\label{thm:revealed_priority}
Suppose that the distributional preference satisfies the upper-bound property and the maximizer property. Then a choice rule $\ch$ is non-wasteful, distributionally maximal, and satisfies the revealed priority axiom if and only if $\ch=\ch^\priority$ for some priority ranking $\priority$.
\end{theorem}

This result complements \autoref{thm:exogenouspriority} by showing that the distributional choice rule does not depend on taking priorities as exogenously fixed. When a priority ranking is given, non-wastefulness, distributional maximality, and freedom from justified envy uniquely identify $\ch^\priority$. When priorities are instead inferred from behavior, non-wastefulness, distributional maximality, and acyclicity of revealed priorities again force the choice rule to be a distributional choice rule for some priority ranking.

Together, Theorems \ref{thm:exogenouspriority}--\ref{thm:revealed_priority} characterize the distributional choice rule at the single-school level. Extending this characterization to centralized markets introduces a new challenge because multiple schools' choice rules, each based on its own priority ranking and distributional preference, must interact through a single matching mechanism. \autoref{sec:two-sided} shows how these structural properties support implementation in centralized matching markets through the deferred-acceptance mechanism.

%===========================

\section{Market Design with Distributional Preferences}\label{sec:two-sided}

This section extends our single-school analysis to centralized matching markets. The market consists of a finite set of schools $\mathcal{C}$ and a finite set of students $\mathcal{S}$. Each school $c\in\mathcal{C}$ has a capacity $q_c\in\mathbb{Z}_{++}$, a priority ranking $\priority_c$ over $\mathcal{S}$, and a distributional preference $\succsim_c$ over subsets of students. Each student $s\in\mathcal{S}$ has a preference relation $P_s$, which is a strict linear order over $\mathcal{C}\cup\{\varnothing\}$, where $\varnothing$ denotes the student's outside option. The corresponding weak order is denoted by $R_s$. We denote the set of all student preference relations by $\mathcal{P}$ and a preference profile by $P\coloneqq(P_s)_{s\in\mathcal{S}}\in\mathcal{P}^\mathcal{S}$. Thus, a matching market is summarized by a tuple $\left(\mathcal{S},\mathcal{C},P,(\succsim_c)_{c\in\mathcal{C}},(\priority_c)_{c\in\mathcal{C}},(q_c)_{c\in\mathcal{C}}\right)$. 

Throughout, school capacities, priority rankings, and distributional preferences are common knowledge. On the other hand, student preferences are private information, which motivates the incentive requirements we impose below.

A matching is a function $\mu:\mathcal{S}\to\mathcal{C}\cup\{\varnothing\}$, where  $\mu(s)\in \mathcal{C}\cup\{\varnothing\}$ represents the school (or outside option) that is matched to $s\in\mathcal{S}$, and $\mu^{-1}(c)\subseteq \mathcal{S}$ represents the subset of students that are matched to school $c\in\mathcal{C}$. We impose that a matching cannot exceed any school's capacity, i.e., $\abs{\mu^{-1}(c)}\leq q_c$ for every $c\in\mathcal{C}$. Let $\mathcal{M}$ denote the set of all matchings.

We say a school $c\in\mathcal{C}$ is acceptable for student $s$ if $c \mathrel{ P_s} \varnothing$.

\begin{axiom}\label{ax:matching_ir}
    A matching $\mu$ is \textbf{individually rational} if every student is matched either to an acceptable school or to the outside option; that is, $\mu(s)\mathrel{ R_s }\varnothing$ for every $s\in\mathcal{S}$.
\end{axiom}

We now adapt the single-school axioms to the matching environment.

\begin{axiom}\label{ax:matching_nonwasteful}
A matching $\mu$ is \textbf{non-wasteful} if, for every $c\in\mathcal{C}$ and $s\in\mathcal{S}$,
\[
c\mathrel{ P_s }\mu(s)\implies \abs{\mu^{-1}(c)}=q_c.
\]
\end{axiom}

Non-wastefulness requires that no school leave a seat empty while a student who wants that seat remains assigned elsewhere or unmatched.

\begin{axiom}\label{ax:matching_diversity}
A matching $\mu$ is \textbf{distributionally maximal} if, for every $c\in\mathcal{C}$ and every $S\subseteq \{s\in\mathcal{S}:c\mathrel{ R_s }\mu(s)\}$, we have
\[
\abs{S}=\abs{\mu^{-1}(c)}\implies S\not\succ_c \mu^{-1}(c).
\]
\end{axiom}

The set $\{s\in\mathcal{S}:c\mathrel{ R_s }\mu(s)\}$ consists of the students who weakly prefer school $c$ to their assignment under $\mu$. In short, this is the set of students who are willing to attend school $c$ under the matching $\mu$. Distributional maximality requires each school to be matched to a set of students that is distributionally undominated among the students who are willing to attend it, holding fixed the number of students assigned to the school.

\begin{axiom}\label{ax:matching_envy}
A matching $\mu$ is \textbf{free of justified envy} if, for every $c\in\mathcal{C}$ and distinct students $s,s'\in\mathcal{S}$ with $\mu(s)=c$ and $c\mathrel{ P_{s'}}\mu(s')$, 
\[
\big(\mu^{-1}(c)\setminus\{s\}\big)\cup\{s'\}\succsim_c \mu^{-1}(c)\implies s\mathrel{\priority_c}s'.
\]
\end{axiom}

This definition integrates distributional objectives with priority rankings. In standard school choice, a student has justified envy if she prefers a school to her assignment and has higher priority than some admitted student. In our framework, the priority claim is conditional on the distributional objective. If replacing an admitted student $s$ with a student $s'$ who wants school $c$ would weakly improve the school's distributional outcome, then $s$ must have higher priority than $s'$. Thus, a lower-priority student can be matched to a school over a higher-priority student only when replacing the former with the latter would fail to weakly improve that school's distributional outcome.

A matching mechanism is a function $\phi:\mathcal{P}^{\mathcal{S}}\to\mathcal{M}$ that maps each profile of student preferences $P$ to a matching $\phi(P)$. We write $\phi_s(P)$ for the assignment of student $s$ under the matching $\phi(P)$. 

We say that a mechanism $\phi$ is non-wasteful if, for every preference profile $P\in\mathcal{P}^\mathcal{S}$, the matching $\phi(P)$ is non-wasteful. We define individually rational, distributionally maximal, and free-of-justified-envy mechanisms analogously.

\begin{axiom}\label{ax:matching_sp}
A mechanism $\phi$ is \textbf{strategy-proof} if no student can benefit from misreporting her preferences: for every $P\in\mathcal{P}^{\mathcal{S}}$, every $s\in\mathcal{S}$, and every $P_s'\in\mathcal{P}$,
\[
\phi_s(P)\mathrel{ R_s }\phi_s(P_s',P_{-s}).
\]
\end{axiom}

We now define the deferred-acceptance mechanism induced by a profile of school choice rules. Let $(\ch_c)_{c\in\mathcal{C}}$ be a profile of choice rules, one for each school.

\begin{quote}
    \textbf{Deferred-Acceptance Algorithm (DA) based on 
    $(\ch_c)_{c\in\mathcal{C}}$}\\
    \textbf{Input:} A profile of student preferences 
    $P=(P_s)_{s\in\mathcal{S}}$.
    \begin{description}[parsep=0pt, itemsep=5pt, leftmargin=1em]
        \item[Step 1:] Each student $s\in\mathcal{S}$ proposes to their most 
            preferred acceptable school, if any. Let $S_1^c$ be the set of proposals 
            received by school $c$. School $c$ tentatively holds the set 
            $\ch_c(S_1^c)$ and permanently rejects 
            $S_1^c\setminus\ch_c(S_1^c)$. If no students are rejected, the 
            algorithm terminates.

        \item[Step $k \geq 2$:] Each student $s\in\mathcal{S}$ rejected at 
            Step $k-1$ proposes to their most preferred acceptable school from 
            which they have not yet been rejected, if any. Let $S_k^c$ be the union of 
            new proposals to $c$ and the students tentatively held by $c$ from 
            the previous step. School $c$ tentatively holds $\ch_c(S_k^c)$ 
            and permanently rejects $S_k^c\setminus\ch_c(S_k^c)$. If no 
            students are rejected, the algorithm terminates.

        \item[Output:] The final tentative assignments constitute the matching 
            outcome.
    \end{description}
\end{quote}

The \textbf{DA matching mechanism} based on $(\ch_c)_{c\in\mathcal{C}}$ is the mechanism whose outcome is determined by the algorithm above for any given preference profile. In the market generated by distributional preferences, the relevant choice-rule profile is $(\ch_c^{\priority_c})_{c\in\mathcal{C}}$, where each school uses its distributional choice rule.

\begin{theorem}\label{thm:da}
Suppose that, for every school $c\in\mathcal{C}$, the distributional preference $\succsim_c$ satisfies the upper-bound property, the maximizer property, and the coherence property. A matching mechanism is non-wasteful, distributionally maximal, individually rational, strategy-proof, and free of justified envy if and only if it is the deferred-acceptance mechanism based on the distributional choice rules $(\ch_c^{\priority_c})_{c\in\mathcal{C}}$.
\end{theorem}

\autoref{thm:da} is the mechanism-level counterpart of the choice-theoretic results in \autoref{sec:choice}. It shows that the same structural properties---upper-bound, maximizer, and coherence---that characterize single-school choice also characterize implementation in centralized markets. If each school uses its distributional choice rule, deferred acceptance produces matchings that are individually rational, non-wasteful, distributionally maximal, and free of justified envy, while preserving strategy-proofness for students. Conversely, any mechanism satisfying these desiderata must coincide with the deferred-acceptance mechanism based on the schools' distributional choice rules. Thus, the upper-bound, maximizer, and coherence properties provide a unified foundation for centralized implementation of distributional objectives.

\section{Applications}\label{sec:applications}

In this section, we apply our framework to several classes of distributional preferences. The purpose is twofold. First, we show how familiar policies, including reserves and floors-and-ceilings policies, can be written as choice rules derived from distributional preferences that satisfy the upper-bound, maximizer, and coherence properties. Hence, these familiar policies are instances of our distributional choice rule for particular distributional preferences. Second, we show how the same framework accommodates distributional objectives that are not naturally expressed as reserves, especially objectives involving overlapping or intersectional identities.

We focus on sets $S\subseteq\mathcal{S}$ with $\abs{S}=q$ without loss of generality. The upper-bound, maximizer, and coherence properties are defined for $q$-sized sets, and thus, it is sufficient to specify how a school's distributional preference ranks sets of size $q$. In particular, two distributional preferences that agree on all comparisons among size-$q$ sets induce the same frontier sets and the same distributional choice rule, even if they disagree on sets of other cardinalities.

It is useful to formalize the attributes that determine a school's distributional preferences. To that end, let $\mathcal{T}\coloneqq\{t_1,\ldots,t_K\}$ denote a finite set of distributionally relevant identities or protected traits. Each student $s\in\mathcal{S}$ has a type $\tau(s)\in\{0,1\}^K$, which is a vector with $\tau_k(s)=1$ if student $s$ has trait $t_k\in\mathcal{T}$ and $\tau_k(s)=0$ otherwise. In \autoref{example:running}, $K=2$ with $t_1=$ minority status and $t_2=$ low SES, so $\tau(s_1)=(1,1)$, $\tau(s_2)=(0,0)$, $\tau(s_3)=(0,1)$, and $\tau(s_4)=(1,0)$. The quantity $\norm{\tau(s)}=\sum_k \tau_k(s)$ captures how many traits a student $s$ has, and hence represents the student's intersectionality.

\subsection{Score-Based Preferences}\label{subsec:score}

We begin with preferences derived from student-level scores, which is a useful application of our framework to both separable and nonseparable distributional preferences. 

Let $f:\mathcal{S}\to\mathbb{R}$ assign a distributional score to each student. Examples of distributional scores include:
\begin{itemize}[parsep=0pt, listparindent=1em, itemsep=0pt]
\item \textbf{Trait-based scoring:} $f(s)=\sigma^{t}(\tau(s))$, where $\sigma^{t}:\{0,1\}^K\to \mathbb{R}$ maps a student's type vector into a score. For example, given weights $\lambda=(\lambda_k)_{k=1}^K\in [0,1]^K$ satisfying $\sum_k\lambda_k=1$ and a type vector $\tau\in\{0,1\}^K$, the trait-based score could be $\sigma^t(\tau)=\lambda\cdot\tau$. Hence, $\lambda_k$ measures the marginal contribution of trait $t_k$ to a student's overall score and captures how much the school values that trait.
\item \textbf{Merit-based scoring:} $f(s)=\sigma^m(s)$, where $\sigma^m(s)\in\mathbb{R}$ is an exam score of student $s$.

\item \textbf{Multidimensional scoring:} $f(s)=g(\sigma^t(\tau(s)),\sigma^m(s))$, where the aggregator function $g$ allows for arbitrary interactions between trait- and merit-based scores.
\end{itemize}

For each set $S\subseteq \mathcal{S}$ with $\abs{S}=q$, label the students in $S=\{s_1,\ldots,s_q\}$ so that $f(s_1)\geq\cdots\geq f(s_q)$. With some abuse of notation, define $f(S)\coloneqq (f(s_1), \ldots, f(s_q))$ as the profile of distributional scores sorted in descending order. We consider two different distributional preferences that compare any two sets of size $q$ based on their respective profiles of scores. 

The first distributional preference, which we call the \emph{index distributional preference}, compares sets by applying an index to their score profiles. Formally, let $\Psi:\mathbb{R}^q\to\mathbb{R}$ be a \emph{diversity index}. We assume that $\Psi$ is strictly increasing in each coordinate, so that replacing one student with a higher-scoring student strictly improves the index. 

Given two sets $S,S'\subseteq\mathcal{S}$ with $\abs{S}=\abs{S'}=q$, we then define the index distributional preference by
\[
S\succsim S' \quad \Longleftrightarrow \quad \Psi(f(S))\geq \Psi(f(S')).
\]
Note that the index distributional preference is complete. Hence, the upper-bound property holds vacuously. We show that this class of distributional preference satisfies all three structural properties. 

\begin{proposition}\label{prop:index}
Suppose the index $\Psi$ is strictly increasing. Then the index distributional preference satisfies the upper-bound property, the maximizer property, and the coherence property.
\end{proposition}

This class of preferences accommodates the \textit{additive distributional preference} in \autoref{example:2} by specifying
\[
\Psi(f(S))=\sum_{s\in S} f(s).
\]
Since $\Psi$ in this case is additively separable, the value of admitting a student is independent of which other students have already been admitted.

More generally, $\Psi$ may be nonseparable, so that the marginal value of a student's score can depend on the scores of the other admitted students. For example, suppose $f(s)\geq 0$ for all $s\in\mathcal{S}$. Given a set $S\subseteq\mathcal{S}$, let $F(\cdot\vert S):\mathbb{R}_+\to [0,1]$ represent its \emph{score survival function} defined as 
\[
F(m\vert S)\coloneqq \frac{\abs{\{s\in S:f(s)\geq m\}}}{q}
\]
for each $m\in\mathbb{R}_+$. Finally, let the index be given by
\[
\Psi(f(S))= \int_0^\infty g(m, F(m|S))dm,
\]
where $g:\mathbb{R}_+\times [0,1]\to\mathbb{R}$ is measurable, strictly increasing in its second argument, satisfies $g(m,0)=0$ for every $m\in\mathbb{R}_+$, and is such that the integral is finite for every $S\subseteq\mathcal S$. Different specifications of the function $g$ correspond to different distributional preferences. The dependence of $g$ on $m$ determines which score thresholds the school values most: placing greater weight on lower thresholds rewards broad representation, while placing greater weight on higher thresholds rewards cohorts containing students with especially high scores. The shape of $g$ in its second argument determines how the school values the fraction of the cohort whose scores exceed each threshold. For example, concavity generates diminishing returns to additional students above a given threshold, while convexity places increasing value on having a larger share of the cohort above that threshold. Since the index depends on the entire survival function of the cohort, it can be nonseparable across students even though it is derived from individual scores.

Such nonseparability is often difficult to accommodate in matching markets because it can conflict with substitutability, and hence with the stability and implementation properties required of centralized mechanisms \citep{echyen12}. Here, however, strict monotonicity of the index implies that the frontier consists of the $q$-student sets with the highest feasible score profile. Thus, this class of preferences preserves the exchange and coherence structure needed for implementation on the frontier while allowing the shape of the index to affect how non-frontier sets are ranked.

The second distributional preference, which we call the \emph{pointwise distributional preference}, compares sets based on coordinate-wise domination of the profile of scores. Formally, given $S,S'\subseteq\mathcal{S}$ with $\abs{S}=\abs{S'}=q$, we define the pointwise distributional preference by
\[
S\succsim S'\quad\Longleftrightarrow\quad f(S)\geq f(S'),
\]
where for $x,y\in\mathbb{R}^q$, $x\geq y$ if $x_i\geq y_i$ for all $i\in\{1,\ldots,q\}$.

Unlike the index distributional preference, the pointwise distributional preference is incomplete. Thus, the upper-bound property is substantive here.

\begin{proposition}\label{prop:pointwise}
The pointwise distributional preference satisfies the upper-bound property, the maximizer property, and the coherence property.
\end{proposition}

While the index and pointwise distributional preferences differ as binary relations, they induce identical frontiers. If $\abs{S}\leq q$, then $\front(S)=\{S\}$. If $\abs{S}>q$, then $\front(S)$ consists of all $q$-student subsets of $S$ whose profile of scores is equal to the vector of the $q$ highest distributional scores available in $S$. Thus, the two preferences lead to the same distributional choice rule even though one is complete and the other is not. This illustrates that what matters for our theory is how the preferences shape the frontier. How preferences compare off-frontier sets is not important.

\subsection{Matroidal Distributional Preferences}\label{subsection:matroidal}

We next consider distributional preferences derived from matroid rank functions. This class is useful because many familiar policies, such as standard reserve policies as well as overlapping reserves \citep{sonmez/yenmez:20, sonmez/yenmez:22}, have a natural matroidal structure. In these cases, the rank of a set measures how much of the school's distributional objective can be satisfied by the students in the set.

Let $M=(\mathcal{S},\mathcal I)$ be a matroid on the student set $\mathcal S$, and let $\rank_M:2^{\mathcal S}\to\mathbb Z_+$ denote its rank function: for any $S\subseteq\mathcal S$, $\rank _M(S)$ is the cardinality of the largest independent subset of $S$.\footnote{See \ref{sec:matroid_theory} for a formal definition of the matroidal rank function.} 

Given $S,S'\subseteq\mathcal{S}$ with $\abs{S}=\abs{S'}=q$, we define the \emph{matroidal distributional preference} induced by $M$ by
\[
S\succsim S'\quad\Longleftrightarrow\quad \rank_M(S)\geq \rank_M(S').
\]
In words, the independent sets of the matroid specify which groups of students can be ``counted'' toward meeting the school's distributional objective. The rank $\rank_M(S)$ is the largest number of students in $S$ that can be counted in this sense. Thus, one set is preferred to another when it contains a larger subset of students that can count toward that objective.

\begin{proposition}\label{prop:matroidal}
The matroidal distributional preference satisfies the upper-bound property, the maximizer property, and the coherence property.
\end{proposition}

The proposition shows that matroidal distributional preferences are nested in our framework. This is useful because, even when the matroid is not the primitive object governing the school's distributional preference, it can still be viewed as a compact representation of how the primitive distributional preference compares sets. 

We now present several classes of matroidal distributional preferences.

\subsubsection{Nonintersectional Matroids}
Let us first consider settings in which each student has at most one of the $K$ protected traits. This arises, for example, when
\begin{enumerate}[label={$(\alph*)$}, parsep=0pt, itemsep=2pt]
    \item the protected traits are mutually exclusive, such as $\mathcal{T}=\{\text{in-state},\text{out-of-state}\}$,
    \item each combination of identities is treated as a distinct protected trait, such as $\mathcal{T}=\{\text{minority-only},\text{low-SES-only},\text{minority-and-low-SES}\}$, or 
    \item each student is assigned to one of her possibly multiple protected traits.
\end{enumerate}
We model such settings via partition matroids, which are particularly useful for formalizing standard reserve policies where each student qualifies for the reserved seats associated with at most one protected trait.

For each trait $t_k\in\mathcal T$, let
\[
\mathcal S_k\coloneqq \{s\in\mathcal S:\tau_k(s)=1\}
\]
be the set of students with trait $t_k$. Suppose that these sets are pairwise disjoint, so each student has at most one protected trait. We allow for some students to have no protected traits, so it is possible that $\bigcup_{k=1}^K\mathcal{S}_k\subseteq\mathcal{S}$. 

Let $R_k\in\mathbb Z_+$ denote the number of reserved seats for students with trait $t_k$, with $\sum_k R_k\leq q$. The associated partition matroid has independent sets
\[
\mathcal I=\left\{I\subseteq\bigcup_{k=1}^K\mathcal{S}_k:\abs{I\cap\mathcal{S}_k}\leq R_k \text{ for every }k\in\{1,\ldots,K\}\right\}.
\]
Thus, students with no protected trait do not contribute to the matroid rank. The rank function is
\[
\rank_M(S)=\sum_{k=1}^K\min\left\{\abs{S\cap\mathcal S_k},R_k\right\}.
\]
For example, suppose $\mathcal{T}=\{\text{in-state},\text{out-of-state}\}$, with reserves $R_1=80$ and $R_2=20$. Consider a set $S$ with 90 in-state students, 15 out-of-state students, and 10 international students. Since international status is not one of the protected traits in $\mathcal{T}$, international students do not contribute to the matroid rank. Thus,
$\rank_M(S)=\min\{80,90\}+\min\{20,15\}=95$.

Thus, $\rank_M(S)$ counts how many reserved seats can be filled by students in $S$, and the resulting frontier consists of the non-wasteful sets that maximize matroid rank. The distributional choice rule induced by partition matroids corresponds to a standard reserve policy \citep{hayeyi13,ehayeyi14,echyen12}, which prioritizes sets of students that fill the maximum number of reserved seats. 

In \autoref{example:running}, we considered two scenarios: one in which the SES reserve is processed first, and one in which the minority reserve is processed first. These two scenarios can be represented by partition-matroid distributional preferences after assigning intersectional students to a single protected trait ex-ante. If intersectional students are assigned to the low-SES trait, the induced distributional choice rule corresponds to the scenario in which the SES reserve is processed first. If instead intersectional students are assigned to the minority trait, the induced distributional choice rule corresponds to the scenario in which the minority reserve is processed first.

The example highlights one drawback of partition matroids; they require each student to belong to at most one type. Intersectional students with multiple protected traits must be assigned to a single identity, which fails to capture the full complexity of their background. This motivates the intersectional matroid structures we examine next.

\subsubsection{Matroids for Intersectional Identities}

We next consider different classes of matroidal distributional preferences that accommodate intersectional identities.

\paragraph{Transversal matroids and overlapping reserves.}

When a student has multiple protected traits, she may be eligible for the reserved seats associated with more than one trait. Partition matroids handle this multiple eligibility by assigning the student to one of her identities ex-ante. Here, we instead consider transversal matroids, which formalize overlapping reserves: a student may be eligible ex-ante for the reserved seats of multiple traits, but, once admitted, she counts toward the reserved seats of at most one trait. Such overlapping, or horizontal, reserves are studied by \citet{sonmez/yenmez:20, sonmez/yenmez:22}.

For each protected trait $t_k\in\mathcal T$, recall that $\mathcal S_k$ is the set of students who have that trait, and $R_k$ is the number of reserved seats associated with the trait, with $\sum_k R_k\leq q$. Unlike partition matroids, we no longer require the sets $\mathcal{S}_k$ to be pairwise disjoint. For each trait $t_k$, index the reserved seats by $\ell=1,\ldots,R_k$. Define
\[
\mathcal R\coloneqq \bigcup_{k=1}^K\{(\ell,k):\ell\in\{1,\ldots,R_k\}\}.
\]
A set $I\subseteq\mathcal S$ is an independent set of the transversal matroid if there exists an injective map $\eta:I\to\mathcal R$ such that, whenever $\eta(s)=(\ell,k)$, we have $\tau_k(s)=1$. In words, each student in an independent set can be assigned to a distinct reserve seat for a trait she has. The injectivity of $\eta$ ensures that no two students are assigned to the same seat, while the construction of $\mathcal R$ ensures that at most $R_k$ students can be assigned to the reserved seats for trait $t_k$.

The rank function of this transversal matroid measures the maximum number of reserved seats that can be filled by students in a given set. In particular, whenever an intersectional student is eligible for several reserve seats, the matroid optimizes which seat she fills. The induced matroidal distributional preference therefore compares sets by the maximum number of reserve seats they can fill. This is how the distributional preference in \autoref{example:3} is derived. 

Notice that when the sets $\mathcal S_k$ are disjoint, we recover the independent sets of partition matroids. Hence, partition matroids are a special case of transversal matroids. More broadly, rather than assigning intersectional students to one of their potentially numerous traits ex-ante, the transversal matroid allows intersectional students to remain eligible for multiple reserved seats while ensuring that each admitted student counts toward at most one reserve seat.

\paragraph{Vector matroids.}

Vector matroids capture a distributional objective under which the school values cohorts containing linearly non-redundant combinations of protected traits.

Recall that each student $s$ is associated with a type $\tau(s)\in\{0,1\}^K$, which we view as a vector in $\mathbb{R}^K$. A set $I\subseteq\mathcal S$ is independent in the vector matroid if the vectors $\{\tau(s):s\in I\}$ are linearly independent over $\mathbb{R}$. For a student $s\notin S$, adding $s$ to $S$ increases its rank precisely when
\[
\tau(s)\notin\operatorname{span}_{\mathbb R}
\{\tau(s'):s'\in S\}.
\]
The associated rank function is
\[
\rank_M(S)
=
\dim\left(
\operatorname{span}_{\mathbb R}
\{\tau(s):s\in S\}
\right).
\]
Thus, the induced matroidal distributional preference compares cohorts according to the maximum number of linearly independent trait profiles they contain. Since $\rank_M(S)\leq K$, the school values at most $K$ mutually non-redundant profiles, regardless of cohort size.

In \autoref{example:vector}, the pair $\{s_1,s_2\}$ has rank one because $\tau(s_2)$ is the zero vector and therefore adds no dimension to the span of $\tau(s_1)$. In contrast, the pairs $\{s_1,s_3\}$, $\{s_1,s_4\}$, and $\{s_3,s_4\}$ each have rank two because the two type vectors are linearly independent and span $\mathbb{R}^2$. Thus, unlike the score-based preference in \autoref{example:2}, which assigns an explicit premium to intersectionality, the vector-matroid preference assigns no intrinsic premium to an intersectional student: she increases rank only when her trait profile adds a dimension not already generated by the other students in the cohort.

\paragraph{Graphic matroids.}
In some settings, the school may not value each trait in isolation, but rather care about the linkages across traits. For example, a student with expertise in machine learning and economics may be valued differently from one whose expertise links machine learning to biology, even if both are interdisciplinary students. This objective reflects the idea that bridging gaps between otherwise distinct traits can create value for the group as a whole \citep{burt1992,burt2004}. Graphic matroids capture this structure.

To model this, suppose each student has exactly two protected traits, so that $\norm{\tau(s)}=2$ for every $s\in\mathcal S$. We can then interpret $\mathcal{T}$ as a set of vertices of a multigraph with each student naturally corresponding to an edge between the two traits she has. Formally, if $\tau_k(s)=\tau_{k'}(s)=1$ with $k\neq k'$, define the edge contributed by the student as $e_s=\{t_k,t_{k'}\}$. For a set $S\subseteq\mathcal S$, let
\[
E(S)\coloneqq \{e_s:s\in S\}
\]
denote the corresponding multiset of edges. The multigraph is then $G=(\mathcal T,E(\mathcal S))$.

A set $I\subseteq\mathcal S$ is an independent set of the graphic matroid if the multigraph $(\mathcal T,E(I))$ contains no cycle. Equivalently, the students in $I$ add new links among traits without creating a redundancy. The rank function of this matroid is \citep[Section~1.1]{oxley}
\[
\rank_M(S)=K-c(E(S)),
\]
where $c(E(S))$ is the number of connected components of the graph $(\mathcal T,E(S))$, counting isolated vertices as their own components.

The induced matroidal distributional preference therefore compares sets by the maximum number of non-redundant trait linkages represented by students in the set. Thus, the school may prefer a cohort that links otherwise separate traits over one whose students all connect traits within the same already-connected cluster.

\autoref{example:phd} and \autoref{fig:graphic-matroid-example} illustrate such a graphic-matroid distributional preference. In that example, the cohort
$\{s_1,s_2,s_3\}$ contributes edges that connect all four fields into a single
component, giving
$\rank_M(\{s_1,s_2,s_3\})=4-1=3$.
By contrast, the cohort $\{s_1,s_3,s_4\}$ contributes edges CS--Economics, CS--ML, and Economics--ML, which form a cycle, while leaving Political Science as an isolated component, giving
$\rank_M(\{s_1,s_3,s_4\})=4-2=2$.
The school therefore prefers $\{s_1,s_2,s_3\}$ to $\{s_1,s_3,s_4\}$. In particular, adding $s_2$ rather than $s_4$ to $\{s_1,s_3\}$ is preferable because $s_2$ introduces a bridge to Political Science, while $s_4$ adds another connection within the already connected CS--Economics--ML cluster.

\paragraph{Gammoid matroids.}
In some settings, the school may care not only about which traits are linked, but also about the paths through which those linkages arise. For example, a medical school may treat two students specializing in primary care as distributionally different if one is a first-generation student who attended a public undergraduate university, while the other comes from a rural hometown and worked as a research assistant before applying to medical school. Even though both students are linked to the same specialty, their backgrounds and prior experiences trace different paths to that specialty, and those paths may shape how they contribute within the admitted cohort. Gammoid matroids capture this structure.

To model this, define a directed graph $G=(\mathcal{T},E)$ on the trait space $\mathcal{T}$. Let $\mathcal{T}^o\subseteq \mathcal{T}$ be the set of source nodes, representing student backgrounds, and let $\mathcal{T}^d\subseteq \mathcal{T}$ be the set of terminal nodes, representing career outcomes. We allow for the possibility that some nodes are neither source nor terminal nodes, so $\mathcal{T}^o\cup\mathcal{T}^d\subseteq \mathcal{T}$. Each student $s\in\mathcal{S}$ is associated with a source node in $\mathcal{T}^o$, and the directed graph determines which terminal nodes in $\mathcal{T}^d$ can be reached from that source. If multiple students share the same source node or need to be linked to the same terminal node, one can use student-specific copies of the relevant nodes and enlarge the vertex set to accommodate such copies.

A set $I\subseteq\mathcal S$ is an independent set of the gammoid matroid if there exist $\abs{I}$ vertex-disjoint paths from the source nodes associated with students in $I$ to terminal nodes. For a set $S\subseteq\mathcal S$, the rank function $\rank_M(S)$ equals the maximum number of students in $S$ that can be linked from their source nodes to terminal nodes by vertex-disjoint paths. Thus, the induced matroidal distributional preference compares sets according to the maximum number of directed source-to-terminal linkages that can be represented without redundancy.

\subsection{Floors and Ceilings}\label{subsec:floors_ceilings}

We next consider floors-and-ceilings policies. These policies specify lower and upper bounds on the number of admitted students from each protected trait. They are commonly written directly as choice rules. We show how they can instead be rationalized by distributional preferences over sets of students.

For this subsection, suppose that each student has exactly one protected trait.\footnote{One can allow for students to have \emph{at most} one protected trait by treating students with no protected trait as belonging to a residual unprotected category $t_0$.} For each trait $t_k\in\mathcal{T}$ and each set $S\subseteq\mathcal S$, let
\[
n_k(S)\coloneqq \sum_{s\in S} \tau_k(s)
\]
denote the number of students in $S$ with that trait, and let $n(S)\coloneqq (n_1(S),\ldots,n_K(S))$ denote the set's trait-composition. A floors-and-ceilings policy specifies, for each trait $t_k$, a lower bound $L_k\in\mathbb Z_+$ and an upper bound $U_k\in\mathbb Z_+$, with $L_k\leq U_k$ \citep{ehayeyi14}. We say that a set $S\subseteq\mathcal{S}$ is \emph{feasible} if $L_k\leq n_k(S)\leq U_k$ for each $k\in\{1,\ldots,K\}$. Otherwise, we say $S$ is infeasible.

A natural first attempt is to define a dichotomous distributional preference, which ranks sets solely by whether they are feasible. For sets $S,S'\subseteq\mathcal S$ with $\abs{S}=\abs{S'}=q$, define
\[
S\succsim S'
\quad\Longleftrightarrow\quad
S\text{ is feasible or }S'\text{ is infeasible}.
\]
This captures a school that treats the bounds as hard constraints: all feasible sets are mutually indifferent and strictly dominate all infeasible sets, while all infeasible sets are mutually indifferent to one another, regardless of how severely they violate the constraints.

\begin{proposition}\label{prop:bounds}
The dichotomous distributional preference satisfies the upper-bound property and the maximizer property.
\end{proposition}

The upper-bound property is satisfied because the dichotomous distributional preference is complete. 
The maximizer property follows because all feasible sets are mutually indifferent and admit symmetric exchanges, as do all infeasible sets.

However, this strict bifurcation has an important limitation: the dichotomous distributional preference can fail the coherence property. Consequently, the induced choice rule can violate path independence, as the following example demonstrates.

\begin{example}\label{ex:dichotomous_floors}
Suppose $q=3$ and there are two protected traits, $t_1$ and $t_2$. The bounds are given by $L_1=2$, $L_2=0$, $U_1=U_2=3$. Let $\mathcal{S}=\{s_1,s_2,s_3,s_4,s_5\}$, where $s_4$ and $s_5$ have trait $t_1$, while $s_1$, $s_2$, and $s_3$ have trait $t_2$. Suppose the priority ranking is $s_1\mathrel{\priority}s_2\mathrel{\priority}s_3\mathrel{\priority}s_4\mathrel{\priority}s_5$.

The frontier of the entire applicant pool is 
\[
\front(\mathcal{S})=\{\{s_1,s_4,s_5\}, \{s_2,s_4,s_5\}, \{s_3,s_4,s_5\}\},
\]
so the distributional choice rule yields $\ch^\priority(\mathcal{S})=\{s_1,s_4,s_5\}$.

Now consider the applicant pool after $s_5$ is removed. In this case, the lower bound for trait $t_1$ cannot be satisfied, so all subsets of the applicant pool are infeasible. The frontier of the reduced applicant pool is 
\[
\front(\mathcal{S}\setminus\{s_5\})=\{\{s_1,s_2,s_3\}, \{s_1,s_2,s_4\}, \{s_1,s_3,s_4\},\{s_2,s_3,s_4\}\},
\]
so the distributional choice rule yields $\ch^\priority(\mathcal{S}\setminus\{s_5\})=\{s_1,s_2,s_3\}$. Because $s_4$ is admitted from the larger pool but rejected after $s_5$ is removed, the choice rule violates substitutability, and hence path independence.
\end{example}

The problem is that the dichotomous preference treats all constraint violations as equivalent. A set that violates a bound by one student is ranked the same as a set that violates a bound by 10 students. To recover path independence, the distributional preference must compare violations more finely.

We now define a graded floors-and-ceilings distributional preference. For each trait $t_k$ and each set $S\subseteq\mathcal{S}$, define the error function
\begin{equation}\label{eq:error_function}
    \varepsilon_k(n_k(S))=
\begin{cases}
-\eta(L_k-n_k(S)) & \text{if } n_k(S)<L_k,\\
0 & \text{if } L_k\leq n_k(S)\leq U_k,\\
-(n_k(S)-U_k) & \text{if } n_k(S)>U_k,
\end{cases}
\end{equation}
where $\eta\geq 0$ is a penalty parameter.\footnote{It is possible to further generalize the error function by choosing a parameter $\eta_k\geq 0$ for each $k=1,\ldots, K$.} The error for trait $t_k$ is zero when the number of students with that trait lies within the desired lower and upper bounds. Falling below the floor is penalized at rate $\eta$, while exceeding the ceiling is penalized at rate one. 

For sets $S,S'\subseteq\mathcal S$ with $\abs{S}=\abs{S'}=q$, define the floors-and-ceilings distributional preference by
\[
S\succsim S'
\quad\Longleftrightarrow\quad
\sum_{k=1}^K \varepsilon_k(n_k(S))\geq \sum_{k=1}^K \varepsilon_k(n_k(S')).
\]
As with the dichotomous distributional preference, the floors-and-ceilings distributional preference ranks sets that satisfy all floor and ceiling bounds as distributionally maximal. Where the two preferences differ, however, is in how they handle constraint violations. Unlike the dichotomous distributional preference, here, the school treats the bounds as soft constraints, and trades off floor and ceiling violations based on the parameter $\eta$. 

\begin{proposition}\label{prop:floors_ceilings}
For any $\eta\geq 0$, the floors-and-ceilings distributional preference satisfies the upper-bound property, the maximizer property, and the coherence property.
\end{proposition}

Because the floors-and-ceilings distributional preference satisfies all three structural properties, the results above imply that the induced distributional choice rule satisfies all four choice axioms, in particular path independence. When $\eta>q$, the induced distributional choice rule closely parallels the standard floors-and-ceilings choice rule \citep{ehayeyi14}: under this parameterization, a single floor violation is more costly than any possible collection of ceiling violations within a $q$-student set, so the school first minimizes floor violations, and then, subject to that, minimizes ceiling violations. 

Returning to \autoref{ex:dichotomous_floors}, consider the floors-and-ceilings distributional preference and suppose $\eta>0$. Notice that the ceiling bounds are never violated, so an error arises only when the floor bounds are violated. In the original applicant pool $\mathcal S$, the frontier is unchanged:
\[
\front(\mathcal{S})=\{\{s_1,s_4,s_5\}, \{s_2,s_4,s_5\}, \{s_3,s_4,s_5\}\},
\]
since these are exactly the sets with total error of zero. Thus, the distributional choice rule again selects $\ch^\priority(\mathcal S)=\{s_1,s_4,s_5\}$. The difference appears after $s_5$ is removed. In the reduced pool $\mathcal S\setminus\{s_5\}$, no set can satisfy the lower bound $L_1=2$, but any set containing $s_4$ has one unit of floor violation, while any set omitting $s_4$ has two units of floor violation. Hence the frontier of the reduced pool is
\[
\front(\mathcal S\setminus\{s_5\})=\{\{s_1,s_2,s_4\},\{s_1,s_3,s_4\},\{s_2,s_3,s_4\}\},
\]
and the distributional choice rule selects $\ch^\priority(\mathcal S\setminus\{s_5\})=\{s_1,s_2,s_4\}$. Thus, $s_4$ remains admitted after $s_5$ is removed. The path-independence failure in \autoref{ex:dichotomous_floors} disappears because the floors-and-ceilings preference distinguishes between different infeasible sets by preferring sets that come closer to meeting the violated bound.

\subsection{Discrete Concavity}\label{subsec:discrete_concavity}
We conclude the application section by connecting the structural properties of distributional preferences to discrete convex analysis. The previous examples were defined through score indices, particular policies, or matroidal representations. Here, we instead begin with a value function over a cohort's trait-composition and impose a discrete exchange condition on that value function. This gives a general way to verify the maximizer and coherence properties for distributional preferences that depend on the composition of the admitted cohort.

Suppose each student belongs to exactly one category in $\mathcal T$. For each set $S\subseteq\mathcal S$, recall $n(S)\coloneqq (n_1(S),\ldots,n_K(S))$ denotes the set's trait-composition. Let $\Phi:\mathbb Z_+^K\to\mathbb R$ be a distributional value function over trait-compositions. For sets $S,S'\subseteq\mathcal S$ with $\abs{S}=\abs{S'}=q$, define the induced distributional preference by
\[
S\succsim S'\quad\Longleftrightarrow\quad \Phi(n(S))\geq \Phi(n(S')).
\]
Since this preference is represented by a real-valued function on $q$-student sets, it is complete. Hence, the upper-bound property holds vacuously. The substantive restrictions are the maximizer and coherence properties, which require the value function $\Phi$ to have an appropriate exchange structure. To that end, we employ an ordinal version of $M$-concavity that is adapted to sets of fixed cardinality. Let $e_k$ denote the $k$-th standard basis vector in $\mathbb Z^K$.

\begin{definition}\label{def:ordinal-exchange}
The value function $\Phi:\mathbb Z_+^K\to\mathbb R$ satisfies \textbf{$q$-ordinal concavity} if, for any trait-compositions $x,y\in\mathbb Z_+^K$ with $\sum_{k=1}^K x_k=\sum_{k=1}^K y_k=q$, and for every $k\in\{1,\ldots, K\}$ such that $x_k>y_k$, there exists $k'\in\{1,\ldots, K\}$ with $x_{k'}<y_{k'}$ such that at least one of the following holds:
\begin{enumerate}[label={$(\roman*)$}, parsep=0pt, itemsep=2pt]
    \item $\Phi(x-e_k+e_{k'})>\Phi(x)$,
    \item $\Phi(y+e_k-e_{k'})>\Phi(y)$, or
    \item $\Phi(x-e_k+e_{k'})=\Phi(x)$ and $\Phi(y+e_k-e_{k'})=\Phi(y)$.
\end{enumerate}
\end{definition}
In economic terms, this condition is a discrete concavity requirement: the school cannot strictly prefer both of two cohorts to the intermediate compositions obtained by moving them one step toward each other.\footnote{This definition is adapted from the ordinal concavity studied by \citet{hakoyeyo2022} and restricted to fixed-cardinality sets.} When one cohort has more students with trait $t_k$ and fewer students with trait $t_{k'}$ than another, exchanging these traits moves each cohort's composition one step toward the other. The $q$-ordinal concavity condition requires that such an exchange either strictly improves the value of one cohort or preserves the value of both.

\begin{proposition}\label{prop:ordinal_concavity}
Suppose that $\Phi$ satisfies $q$-ordinal concavity. Then the distributional preference induced by $\Phi$ satisfies the upper-bound property, the maximizer property, and the coherence property.
\end{proposition}

A simple and important class is separable discrete-concave value functions
\[
\Phi(x)=\sum_{k=1}^K \phi_k(x_k),
\]
where, for each $k\in\{1,\ldots,K\}$, the function $\phi_k:\mathbb Z_+\to\mathbb R$ has weakly decreasing marginal increments, i.e., $\phi_k(x_k+1)-\phi_k(x_k)$ is weakly decreasing in $x_k$.\footnote{In fact, when each $\phi_k$ has weakly decreasing marginal increments, the separable value function $\Phi$ satisfies the fixed-cardinality exchange inequality associated with $M$-concavity. This cardinal exchange condition implies $q$-ordinal concavity.} Since the marginal value of admitting an additional student of a given trait decreases as that trait becomes more represented in the cohort, such a separable value function $\Phi$ satisfies $q$-ordinal concavity. 

Notice that the floors-and-ceilings preference in the previous section is a special case of a preference induced by a $q$-ordinally concave value function where each $\phi_k$ is given by \eqref{eq:error_function}. In that case, the marginal increments are given by 
\[
\phi_k(x_k+1)-\phi_k(x_k)=\begin{cases}
    \eta & \mbox{ if } x_k<L_k\\
    0 & \mbox{ if } L_k\leq x_k<U_k\\ 
    -1 & \mbox{ if } x_k\geq U_k,
\end{cases}
\]
which is indeed weakly decreasing in $x_k$. Thus, \autoref{prop:ordinal_concavity} immediately implies \autoref{prop:floors_ceilings}.   

More broadly, $q$-ordinal concavity imposes the same exchange logic on trait-compositions that the maximizer and coherence properties impose on the exchange of students across cohorts. Hence, unlike the score-based preferences in \autoref{subsec:score}, which rely on monotonicity properties of the index function, this class relies on local exchange properties of the value function. It therefore provides an alternative route for verifying that a distributional preference satisfies the three structural properties needed for implementation.

\section{Conclusion}\label{sec:conclusion}

This paper develops a preference-based approach to distributional objectives in school choice and related selection problems. Rather than encoding a school's distributional objectives directly through reserves, quotas, or institutional choice rules, we take as primitive the school's distributional preference over cohorts. This approach accommodates overlapping and intersectional identities and makes it possible to ask when a distributional objective can be reconciled with priority rankings, consistent choice across applicant pools, and centralized matching.

We identify an upper-bound property and two exchange properties. The upper-bound property ensures that, within each applicant pool, all distributionally maximal non-wasteful cohorts are indifferent and strictly dominate the remaining non-wasteful cohorts. Conditional on this property, the maximizer property holds if and only if, for every priority ranking, the greedy distributional choice rule is the unique choice rule satisfying non-wastefulness, distributional maximality, and freedom from justified envy. Conditional on the same property, the maximizer and coherence properties hold if and only if, for every priority ranking, that rule satisfies path independence. When every school's distributional preference satisfies all three structural properties, deferred acceptance based on these rules is the unique mechanism satisfying the corresponding three axioms together with individual rationality and strategy-proofness.

These results also delineate where the characterizations fail. Conditional on the upper-bound property, failure of the maximizer property implies that there exists a priority ranking for which another choice rule satisfies the three choice axioms, while failure of either exchange property implies that there exists a priority ranking for which the distributional choice rule fails path independence. Such failures undermine the choice-theoretic structure supporting the deferred-acceptance characterization.

The applications show that both familiar and less conventional distributional objectives fit within the framework. Standard and overlapping reserves arise from matroidal distributional preferences, while a graded formulation of floors and ceilings also satisfies the three structural properties. The framework additionally accommodates score-based preferences that assign different values to intersectional profiles, preferences that value linearly non-redundant combinations of traits, and preferences that value linkages across traits. It therefore provides preference foundations for canonical policies while supplying structured formulations for distributional objectives beyond conventional reserve systems.

Our analysis suggests a preference-first approach to policy design. Rather than selecting a reserve policy and subsequently investigating its properties, a designer can begin with the institution's distributional preference and ask whether it satisfies the structural properties supporting the desired choice and matching guarantees. The framework thereby separates the normative question of which cohorts are distributionally desirable from the implementation question of which objectives can be reconciled with individual priorities and centralized matching.

\clearpage
\appendix
\addcontentsline{toc}{section}{Appendices}
\renewcommand{\thesubsection}{Appendix \Alph{subsection}}

\section*{Appendix}

\subsection{Useful Results from Matroid Theory}\label{sec:matroid_theory}

We first state some equivalent characterizations for the bases of a matroid, which we use in our proofs. Fix a ground set $\mathcal{E}$. A collection of sets $\mathcal{B}\subseteq 2^{\mathcal{E}}$ forms the bases of a matroid if and only if the following hold:
\begin{itemize}[parsep=0pt, listparindent=1em, itemsep=0pt]
    \item[(B1)] \textit{Non-emptiness}: $\mathcal{B}\neq\varnothing$.
    \item[(B2)] \textit{Exchange}: For every $B,B'\in\mathcal{B}$ and $e\in B\setminus B'$, there exists $e'\in B'\setminus B$ such that $(B\setminus\{e\})\cup\{e'\}\in\mathcal{B}$.
\end{itemize}
One may replace (B2) with the following \emph{strong simultaneous exchange} property \citep[Lemma~2.2.2]{oxley}: for any $B,B'\in\mathcal{B}$ and $e\in B\setminus B'$, there exists $e'\in B'\setminus B$ such that  
\[
(B\setminus\{e\})\cup\{e'\}\in \mathcal{B} \quad \text{ and } \quad (B'\setminus\{e'\})\cup\{e\}\in\mathcal{B}.
\]

Similarly, one may replace the exchange axiom (B2) with the following \emph{weak simultaneous exchange} property \citep[Theorem 4.3]{Murota:SIAM:2003}: for any distinct $B,B'\in\mathcal{B}$, there exists $e\in B\setminus B'$ and $e'\in B'\setminus B$ such that 
\[
(B\setminus\{e\})\cup\{e'\}\in \mathcal{B} \quad \text{ and } \quad (B'\setminus\{e'\})\cup\{e\}\in\mathcal{B}.
\]

Given a matroid $M=(\mathcal{E},\mathcal{I})$, where $\mathcal{I}$ is its collection of independent sets, and a priority ranking $\priority$ on $\mathcal{E}$, the \textbf{standard matroid greedy algorithm} initializes $G=\varnothing$ and scans the elements of $\mathcal{E}$ in $\priority$-order. When an element $e$ is considered, the algorithm adds $e$ to $G$ if and only if $G\cup\{e\}\in\mathcal{I}$. The algorithm terminates with a maximal independent set, hence a basis. We call this set the \textbf{greedy basis} of $M$ under $\priority$ and denote it by $G_{M}^{\priority}$. We use the following classic result regarding greedy algorithms on matroids.

\begin{lemma}[\citealp*{gale1968}]\label{lem:Gale}
Let $\mathcal{E}$ be a ground set, let $M=(\mathcal{E},\mathcal{I})$ be a matroid on $\mathcal{E}$, and let $\priority$ be a priority ranking over $\mathcal{E}$. Then $G^\priority_M$ priority dominates every $I\in\mathcal{I}$.
\end{lemma}

Given a matroid $M=(\mathcal{E},\, \mathcal{I})$, its rank function $\rank_M:2^\mathcal{E}\to\mathbb Z_+$ is defined by
\[
\rank_M(E)\coloneqq \max_{I\subseteq E,\ I\in\mathcal I}\abs{I} 
\]
for each $E\subseteq \mathcal{E}$. We refer to $\rank_M(\mathcal{E})$ as the \textbf{matroid rank} of $M$.

We will use a quotient relation between two matroids on a common ground set. In our application, the relevant quotients have ranks that differ by one. Given two matroids $M^{c}$ and $M^{d}$ on a finite common ground set $\mathcal{E}$, we say that $M^{c}$ is a \textbf{quotient} of $M^{d}$ if for every  $E'\subseteq E\subseteq \mathcal{E}$, 
\begin{equation}\label{eq:quotient-rank}
\rank_{M^c}(E)-\rank_{M^c}(E')\leq  \rank_{M^d}(E)-\rank_{M^d}(E').
\end{equation}
Furthermore, if $\rank_{M^d}(\mathcal{E})=\rank_{M^c}(\mathcal{E})+1$, we say that $M^c$ is a \textbf{rank-one quotient} of $M^d$.

There is an equivalent characterization of rank-one quotients via an elementary construction using deletion and contraction operations on matroids. Formally, given a matroid $M=(\mathcal E,\mathcal I)$ and an element $e\in \mathcal{E}$, let 
\[
\mathcal{I}^{\setminus e}\coloneqq \{I\in \mathcal{I}: I\subseteq \mathcal{E}\setminus\{e\}\},
\]
and for $e$ satisfying $\{e\}\in\mathcal{I}$, let
\[
\mathcal{I}^{/e}\coloneqq \{I\subseteq \mathcal{E}\setminus\{e\} : I\cup \{e\}\in\mathcal I\}.
\]
The first collection consists of the independent sets of $M$ that do not contain $e$. The second consists of the sets obtained by taking an independent set of $M$ that contains $e$ and removing $e$. Both $\mathcal{I}^{\setminus e}$ and $\mathcal{I}^{/e}$ define a collection of independent sets on the reduced ground set $\mathcal{E}\setminus\{e\}$. Moreover, by the hereditary property, $\mathcal{I}^{/e}\subseteq \mathcal{I}^{\setminus e}$. 

The matroid $M^{\setminus e}\coloneqq(\mathcal{E}\setminus \{e\},\, \mathcal{I}^{\setminus e})$ is the \textbf{deletion of $e$ from $M$}, while the matroid $M^{/e}\coloneqq (\mathcal{E}\setminus \{e\},\, \mathcal{I}^{/e})$ is the \textbf{contraction of $M$ by $e$}.

Given two matroids $M^{c}$ and $M^{d}$ on a finite common ground set $\mathcal{E}$, $M^{c}$ is a rank-one quotient of $M^{d}$ if and only if there exists an element $e\notin\mathcal E$ and a matroid $M=(\mathcal E\cup\{e\},\mathcal I)$ with $\{e\}\in\mathcal I$ such that $M^{c}=M^{/e}$ is the contraction of $M$ by $e$ (hence ``c'' for contraction) and $M^{d}=M^{\setminus e}$ is the deletion of $e$ from $M$ (hence ``d'' for deletion) \citep[Section~7.3]{oxley}. 

The following two auxiliary facts are standard consequences of the theory of matroid quotients and flag matroids. We include direct proofs because the particular formulations used here are tailored to our application.

\begin{lemma}\label{lem:nested-greedy-quotient}% [Nested-greedy characterization of rank-one quotients]
Let $M^c$ and $M^d$ be matroids on the same finite ground set $\mathcal{E}$. Then $M^c$ is a quotient of $M^d$ if and only if $G_{M^c}^\priority\subseteq G_{M^d}^\priority$ for every priority ranking $\priority$ on $\mathcal{E}$.
\end{lemma}
\begin{proof}
\noindent\textbf{($\Rightarrow$).}
Suppose $M^c$ is a quotient of $M^d$. Then \eqref{eq:quotient-rank} holds. Fix a priority ranking $\priority$ on $\mathcal{E}$, and write the elements of $\mathcal{E}$ as
\[
e_1\mathrel{\priority} e_2\mathrel{\priority}\cdots
\mathrel{\priority} e_{\abs{\mathcal{E}}}.
\]
For each $k\in\{1,\ldots,\abs{\mathcal{E}}\}$, let $E_k:=\{e_1,\ldots,e_k\}$ and $E_0:=\varnothing$. For an arbitrary matroid $M$ on $\mathcal{E}$, the standard matroid greedy algorithm admits $e_k$ if and only if $\rank_M(E_k)-\rank_M(E_{k-1})=1$. 

Suppose $e_k\in G_{M^c}^\priority$. Then, applying \eqref{eq:quotient-rank} to $E_{k-1}\subseteq E_k$ and using the fact that a rank function can increase by at most one when a single element is added to a set, we have 
\[
1=\rank_{M^c}(E_k)-\rank_{M^c}(E_{k-1})\leq \rank_{M^d}(E_k)-\rank_{M^d}(E_{k-1})\leq 1.
\]
Thus, $e_k\in G_{M^d}^\priority$. Since the priority ranking was fixed arbitrarily, we conclude that $G_{M^c}^\priority\subseteq G_{M^d}^\priority$ for all $\priority$ on $\mathcal{E}$.

\medskip
\noindent\textbf{($\Leftarrow$).}
Suppose $G_{M^c}^\priority\subseteq G_{M^d}^\priority$ for every priority ranking $\priority$ on $\mathcal E$. We verify \eqref{eq:quotient-rank}. Fix $E'\subseteq E\subseteq \mathcal E$. Choose a priority ranking $\priority$ on $\mathcal E$ that ranks the elements of $E'$ first, then the elements of $E\setminus E'$, and then the elements of $\mathcal E\setminus E$, with arbitrary order within each block.

When restricted to the first block $E'$, greedy on $M^d$ under $\priority$ produces a maximal independent subset of $E'$ in $M^d$ of cardinality $\rank_{M^d}(E')$. After the elements of $E'$ and $E\setminus E'$ have all been processed, the selected elements inside $E$ form a maximal independent subset of $E$ in $M^d$ of cardinality $\rank_{M^d}(E)$. Hence
\[
\abs{G^\priority_{M^d}\cap E'}=\rank_{M^d}(E')\quad\text{and}\quad \abs{G^\priority_{M^d}\cap E}=\rank_{M^d}(E).
\]
Therefore $\abs{G^\priority_{M^d}\cap (E\setminus E')}=\rank_{M^d}(E)-\rank_{M^d}(E')$. The same argument for $M^c$ gives
\[
\abs{G^\priority_{M^c}\cap (E\setminus E')}=\rank_{M^c}(E)-\rank_{M^c}(E').
\]
Since $G^\priority_{M^c}\subseteq G^\priority_{M^d}$, we obtain
\[
\rank_{M^c}(E)-\rank_{M^c}(E')=\abs{G^\priority_{M^c}\cap (E\setminus E')}\leq \abs{G^\priority_{M^d}\cap (E\setminus E')}=\rank_{M^d}(E)-\rank_{M^d}(E').
\]
This verifies \eqref{eq:quotient-rank}. Hence, $M^c$ is a quotient of $M^d$.
\end{proof}

\begin{lemma}\label{lem:nested-greedy-exchange}
Let $M^c$ and $M^d$ be matroids on a finite set $\mathcal E$, and suppose that $M^c$ is a rank-one quotient of $M^d$. Then, for every basis $B^c$ of $M^c$ and every basis $B^d$ of $M^d$ with $B^c\setminus B^d\neq\varnothing$, there exist $e^c\in B^c\setminus B^d$ and $e^d\in B^d\setminus B^c$ such that $(B^c\setminus\{e^c\})\cup\{e^d\}$ is a basis of $M^c$ and
$(B^d\setminus\{e^d\})\cup\{e^c\}$ is a basis of $M^d$.
\end{lemma}

\begin{proof}[Proof of \autoref{lem:nested-greedy-exchange}]
By the deletion-contraction characterization of rank-one quotients, there exists an element $e\notin \mathcal E$ and a matroid ${M}=(\mathcal{E}\cup \{e\},\, \mathcal{I})$ with $\{e\}\in\mathcal{I}$ such that $M^c=M^{/e}$ and $M^d=M^{\setminus e}$. Then $B^c$ is a basis of $M^c$ if and only if $B^c\cup\{e\}$ is a basis of $M$, and $B^d$ is a basis of $M^d$ if and only if $B^d$ is a basis of $M$ and $e\notin B^d$. 

Let $B^c$ be a basis of $M^c$ and let $B^d$ be a basis of $M^d$ satisfying $B^c\setminus B^d\neq\varnothing$. Take any $e^c\in B^c\setminus B^d$. Then, $B^c\cup\{e\}$ and $B^d$ are each a basis of ${M}$, and $e^c\in (B^c\cup\{e\})\setminus B^d$. By the strong basis exchange property \citep[e.g.,][Lemma~2.2.2]{oxley}, there exists $e^d\in B^d\setminus (B^c\cup\{e\})=B^d\setminus B^c$ such that both
\[
\bigl((B^c\cup\{e\})\setminus\{e^c\}\bigr)\cup\{e^d\}\quad \text{ and } \quad (B^d\setminus\{e^d\})\cup\{e^c\}
\]
are each a basis of $M$ as well.

The first basis contains $e$, so contracting $e$ from it gives $(B^c\setminus\{e^c\})\cup\{e^d\}$, which is a basis of $M^c$. The second basis does not contain $e$, so it is a basis of $M^d$. This proves the lemma.
\end{proof}

\subsection{Proofs of Lemmas and Propositions}

\begin{proof}[Proof of \autoref{lem:distfrontier}]\

\noindent\textbf{($\Leftarrow$).} 
Suppose that $\ch(S)\in\front(S)$ for each set $S\subseteq\mathcal{S}$. By the definition of the frontier, $\front(S)\subseteq\nw(S)$, so $\ch(S)\in\nw(S)$. Hence $\abs{\ch(S)}=\min\{q,\abs{S}\}$, and $\ch$ is non-wasteful. Moreover, any set $S'\subseteq S$ with $\abs{S'}=\abs{\ch(S)}$ is also non-wasteful, so $S'\in \nw(S)$. By definition of the frontier and the fact that $\ch(S)\in \front(S)$, we conclude that $S'\not\succ \ch(S)$. Thus, $\ch$ is also distributionally maximal. \medskip

\noindent\textbf{($\Rightarrow$).} Suppose that $\ch$ is non-wasteful and distributionally maximal. Consider any set $S\subseteq\mathcal{S}$. Non-wastefulness implies $\abs{\ch(S)}=\min\{q,\abs{S}\}$, so $\ch(S)\in\nw(S)$. Since $\ch$ is distributionally maximal, there is no $S'\subseteq S$ with $\abs{S'}=\abs{\ch(S)}$ such that $S'\succ\ch(S)$. Equivalently, there is no $S'\in\nw(S)$ such that $S'\succ\ch(S)$. Hence $\ch(S)\in\front(S)$.
\end{proof}

\noindent\begin{proof}[Proof of \autoref{prop:comp}]\ 

\noindent\textbf{($\Rightarrow$).} Suppose the distributional preference satisfies the upper-bound property. Consider any set $S\subseteq\mathcal{S}$. If $\abs{S}\le q$, then $\nw(S)=\front(S)=\{S\}$, so Part $(i)$ of the proposition holds trivially. Moreover, $\nw(S)\setminus \front(S)=\varnothing$, so Part $(ii)$ of the proposition holds vacuously. Hence, for the remainder of the proof, suppose $\abs{S}>q$.\medskip 

\noindent\textbf{Part $(i)$:}  Let $S',S''\in\front(S)$. Toward a contradiction, suppose that $S'$ and $S''$ are incomparable. By the upper-bound property, there exists $T\in\nw(S'\cup S'')$ such that $T\succ S'$ or $T\succ S''$. Since $S'\cup S''\subseteq S$ and $\abs{S'\cup S''}\geq q$, we have $T\in\nw(S)$. This contradicts either $S'\in\front(S)$ or $S''\in\front(S)$. Thus $S'$ and $S''$ must be comparable. Additionally, because neither strictly dominates the other, $S'\sim S''$.\medskip

\noindent\textbf{Part $(ii)$:}
Let $S'\in\front(S)$ and let 
$S''\in\nw(S)\setminus\front(S)$. Since $S''\notin\front(S)$, there exists  $T_1\in\nw(S)$ such that $T_1\succ S''$. If $T_1\in\front(S)$, then by Part $(i)$ of the proposition, $S'\sim T_1$, and therefore $S'\succ S''$ by transitivity. 

If $T_1\notin\front(S)$, then we can repeat the above argument to obtain  $T_2\in\nw(S)$ with $T_2\succ T_1$. Iteratively, whenever $T_k\notin\front(S)$, there exists $T_{k+1}\in\nw(S)$ such that $T_{k+1}\succ T_k$. This yields a sequence $T_1, T_2, \ldots$ in $\nw(S)$ such that 
\[
T_k\succ T_{k-1}\succ\ldots\succ T_1\succ S''.
\]
Since $\nw(S)$ is a finite set, and the strict part of a preorder is transitive and irreflexive, the sequence construction cannot continue indefinitely. Hence, there exists some $K$ with $T_K\in\front(S)$; otherwise the sequence construction could be continued. By construction, $T_K\succ S''$, and by Part $(i)$ of the proposition, $S'\sim T_K$. Therefore, by transitivity, $S'\succ S''$.\medskip

\noindent\textbf{($\Leftarrow$).} Suppose Parts $(i)$ and $(ii)$ of the proposition hold. Consider any two sets $S,S'\subseteq\mathcal{S}$ such that $S$ and $S'$ are incomparable and that $\abs{S}=\abs{S'}=q$. Hence, the antecedent for the upper-bound property holds.

We first show that $S\notin \front(S\cup S')$ and $S'\notin\front(S \cup S')$. Indeed, if both were in $\front(S\cup S')$, then Part $(i)$ would imply that $S\sim S'$, contradicting incomparability. Similarly, if only one of them is in $\front(S\cup S')$, say $S\in\front(S\cup S')$ and $S'\notin\front(S\cup S')$, then Part $(ii)$ would imply that $S\succ S'$, again contradicting incomparability. Hence, $S, S'\in\nw(S\cup S')\setminus \front(S\cup S')$.

We already know that $\front(S\cup S')$ is non-empty. Moreover, for every $T\in \front(S\cup S')$, Part $(ii)$ implies that $T\succ S$ and $T\succ S'$, which establishes the existence of a set $T\in\nw(S\cup S')$ that dominates either $S$ or $S'$ as desired.
\end{proof}

\begin{proof}[Proof of \autoref{lem:matroidbase}]\

\noindent\textbf{($\Rightarrow$).} Suppose the maximizer property holds. Consider any set $S\subseteq\mathcal{S}$. We shall show that $\front(S)$ satisfies the bases axioms: non-emptiness (B1) and exchange (B2) (see \ref{sec:matroid_theory}).  

By construction, $\front(S)$ is non-empty, so (B1) is satisfied. If $\front(S)$ is a singleton, then the exchange condition is satisfied vacuously. Thus, suppose $\abs{\front(S)}>1$. Then $\abs{S}>q$, so every set in $\front(S)$ has cardinality $q$.

Take two distinct sets $S',S''\in\front(S)$. Since the distributional preference satisfies the upper-bound property, \autoref{prop:comp} implies $S'\sim S''$. Moreover, for every $T\in\nw(S)$, \autoref{prop:comp} implies $S'\succsim T$ and $S''\succsim T$. Since $\nw(S'\cup S'')\subseteq \nw(S)$, we have $S'\succsim T$ and $S''\succsim T$ for all $T\in\nw(S'\cup S'')$. By the maximizer property, there exist $s'\in S'\setminus S''$ and $s''\in S''\setminus S'$ such that
\[
(S'\setminus\{s'\})\cup\{s''\}\sim S' \quad\text{and}\quad (S''\setminus\{s''\})\cup\{s'\}\sim S''.
\]

Both exchanged sets belong to $\nw(S)$. Since each is indifferent to a frontier set, \autoref{prop:comp} implies that both exchanged sets also belong to $\front(S)$. Hence, $\front(S)$ satisfies the weak simultaneous exchange property, which is equivalent to (B2) by \citet[Theorem~4.3]{Murota:SIAM:2003}. Therefore, $\front(S)$ forms the bases of a matroid.\medskip

\noindent\textbf{($\Leftarrow$).}  Consider two distinct sets $S,S'\subseteq \mathcal{S}$ with $\abs{S}=\abs{S'}=q$ such that $S\succsim S''$ and $S'\succsim S''$ for all $S''\in\nw(S\cup S')$. By definition, $S,S'\in\front(S\cup S')$. Since $\front(S\cup S')$ forms the bases of a matroid for ground set $S\cup S'$, it satisfies the weak simultaneous exchange property:  there exist $s\in S\setminus S'$ and $s'\in S'\setminus S$ such that
\[
(S\setminus\{s\})\cup\{s'\}\in\front(S\cup S') \quad\text{and}\quad (S'\setminus\{s'\})\cup\{s\}\in\front(S\cup S').
\]

Since the upper-bound property is satisfied, by \autoref{prop:comp}, all frontier sets are indifferent to each other. Specifically, we have
\[
(S\setminus\{s\})\cup\{s'\}\sim S \quad\text{and}\quad (S'\setminus\{s'\})\cup\{s\}\sim S'.
\]
Hence, $\succsim$ also satisfies the maximizer property. 
\end{proof}

\begin{proof}[Proof of \autoref{prop:index}]
We will use the following observation repeatedly: for any $S\subseteq\mathcal{S}$,  $s\in S$, and $s'\notin S$ with $f(s')>f(s)$, we have $f\big((S\setminus\{s\})\cup\{s'\}\big)> f(S)$, where for $x,y\in\mathbb{R}^q$, $x>y$ if $x_i\geq y_i$ for all $i\in\{1,\ldots,q\}$ and a strict inequality in at least one dimension.

\medskip
\noindent
\textit{Upper-bound Property:} For any $S,S'\subseteq\mathcal{S}$ with $\abs{S}=\abs{S'}=q$, either $\Psi(f(S))\geq \Psi(f(S'))$ so that $S\succsim S'$, or $\Psi(f(S))\leq \Psi(f(S'))$ so that $S'\succsim S$. Since no two sets of size $q$ are incomparable, the upper-bound property holds vacuously. 

\medskip
\noindent
\textit{Maximizer Property:}
Let $S,S'\subseteq\mathcal{S}$ be distinct sets with $\abs{S}=\abs{S'}=q$, and suppose that $S\succsim S''$ and $S'\succsim S''$ for every $S''\in\nw(S\cup S')$.

Take any $s\in S\setminus S'$ and any $s'\in S'\setminus S$. We claim that $f(s)=f(s')$. Suppose first that $f(s)>f(s')$. Then $f\big((S'\setminus\{s'\})\cup\{s\}\big)> f(S')$, and by strict monotonicity of $\Psi$, we get $\Psi\big(f((S'\setminus\{s'\})\cup\{s\})\big)>\Psi( f(S'))$. This implies $(S'\setminus\{s'\})\cup\{s\}\succ S'$, contradicting the assumption that $S'$ weakly dominates every element of $\nw(S\cup S')$. Similarly, if $f(s')>f(s)$, an analogous argument would yield $(S\setminus\{s\})\cup\{s'\}\succ S$, contradicting the assumption that $S$ weakly dominates every element of $\nw(S\cup S')$. 

Therefore, $f(s)=f(s')$. Exchanging equal-score students leaves the profile of scores unchanged for both sets, so 
\[
\Psi\big(f((S'\setminus\{s'\})\cup\{s\})\big)=\Psi( f(S'))\quad \text{ and }\quad \Psi\big(f((S\setminus\{s\})\cup\{s'\})\big)=\Psi(f(S)).
\]
Consequently,
\[
(S'\setminus\{s'\})\cup\{s\}\sim S' \quad\text{and}\quad (S\setminus\{s\})\cup\{s'\}\sim S.
\]
This proves the maximizer property.

\medskip
\noindent
\textit{Coherence Property:}
Let $S,S'\subseteq\mathcal{S}$ be distinct sets with $\abs{S}=\abs{S'}=q$, and let $\hat s\in S\setminus S'$ with $(S\setminus S')\setminus\{\hat s\}\neq\emptyset$. Suppose that the three antecedents in \autoref{def:exchange-coherence} hold. 

Choose any $s\in (S\setminus S')\setminus\{\hat s\}$ and any $s'\in S'\setminus S$. Such an $s'$ exists because $\abs{S}=\abs{S'}$ and $S\setminus S'$ is nonempty, which implies that $S'\setminus S$ is also nonempty. Since $s\neq \hat s$, the set $(S\setminus\{s\})\cup\{s'\}$ contains $\hat s$ and belongs to $\nw(S\cup S')$. By condition $(ii)$, $S\succsim (S\setminus\{s\})\cup\{s'\}$. If $f(s')>f(s)$, then $\Psi\big(f((S\setminus\{s\})\cup\{s'\})\big)>\Psi(f(S))$, contradicting the preceding weak dominance. Hence $f(s)\geq f(s')$.

Similarly, $(S'\setminus\{s'\})\cup\{s\}$ belongs to $\nw((S\cup S')\setminus\{\hat s\})$. By condition $(iii)$, $S'\succsim (S'\setminus\{s'\})\cup\{s\}$. If $f(s)>f(s')$, then $\Psi\big(f((S'\setminus\{s'\})\cup\{s\})\big)>\Psi(f(S'))$, contradicting the preceding weak dominance. Hence $f(s')\ge f(s)$.

Therefore, $f(s)=f(s')$. Exchanging equal-score students leaves the profile of scores unchanged for both sets, so
\[
(S\setminus\{s\})\cup\{s'\}\sim S \quad\text{and}\quad (S'\setminus\{s'\})\cup\{s\}\sim S'.
\]
This proves the coherence property.
\end{proof}

\begin{proof}[Proof of \autoref{prop:pointwise}]\

\medskip
\noindent
\textit{Upper-bound Property:}
Let $S,S'\subseteq\mathcal{S}$ satisfy $\abs{S}=\abs{S'}=q$, and suppose that $S$ and $S'$ are incomparable. Let $T\subseteq S\cup S'$ be a set of $q$ students with the highest $f$-values in $S\cup S'$. Then $T\in\nw(S\cup S')$. Moreover, by construction, $T$ pointwise weakly dominates every set in $\nw(S\cup S')$, and in particular $T\succsim S$ and $T\succsim S'$. If neither comparison were strict, then we would have $T\sim S$ and $T\sim S'$, which would imply $S\sim S'$, contradicting incomparability. Therefore, $T$ strictly dominates either $S$ or $S'$. Hence, the upper-bound property holds.

\medskip
\noindent
\textit{Maximizer and Coherence Properties:} The arguments are analogous to those for \autoref{prop:index}, with pointwise dominance of the ordered vector of $f$-values in place of strict monotonicity of $\Psi$.
\end{proof}

\noindent\begin{proof}[Proof of \autoref{prop:matroidal}] Fix a matroid $M=(\mathcal{S},\mathcal{I})$. We use the following standard exchange property of rank functions: for any two sets $X,Y\subseteq\mathcal{S}$ with $\abs{X}=\abs{Y}$ and any $x\in X\setminus Y$, there exists $y\in Y\setminus X$ such that \begin{equation}\label{eq:rank-exchange}
\rank_M(X)+\rank_M(Y)\leq \rank_M\big((X\setminus\{x\})\cup\{y\}\big)+\rank_M\big((Y\setminus\{y\})\cup\{x\}\big).
\end{equation}
This is the fixed-cardinality exchange property for matroid rank functions, which is equivalent to $\rank_M$ being $M$-concave when restricted to a domain of equi-cardinal sets \citep{shioura2012}.

\noindent
\textit{Upper-bound Property:}
For any $S,S'\subseteq\mathcal{S}$ with $\abs{S}=\abs{S'}=q$, either $\rank_M(S)\geq \rank_M(S')$ so that $S\succsim S'$, or $\rank_M(S')\geq \rank_M(S)$ so that $S'\succsim S$. Since no two sets of size $q$ are incomparable, the upper-bound property holds vacuously.

\medskip
\noindent
\textit{Maximizer Property:}
Let $S,S'\subseteq\mathcal{S}$ be distinct sets with $\abs{S}=\abs{S'}=q$, and suppose that $S\succsim S''$ and $S'\succsim S''$ for every $S''\in\nw(S\cup S')$. Equivalently,
$\rank_M(S)\geq \rank_M(S'')$ and $\rank_M(S')\geq \rank_M(S'')$ for every $S''\in\nw(S\cup S')$.

Choose any $s\in S\setminus S'$. By \eqref{eq:rank-exchange}, there exists $s'\in S'\setminus S$ such that
\[
\rank_M(S)+\rank_M(S')\leq \rank_M\big((S\setminus\{s\})\cup\{s'\}\big)+ \rank_M\big((S'\setminus\{s'\})\cup\{s\}\big).
\]
However, both $S$ and $S'$ maximize the rank function over $\nw(S\cup S')$, and both exchanged sets belong to $\nw(S\cup S')$. Consequently, the above inequality must bind, and 
\[
\rank_M\big((S\setminus\{s\})\cup\{s'\}\big)=\rank_M(S) \quad\text{and}\quad \rank_M\big((S'\setminus\{s'\})\cup\{s\}\big)=\rank_M(S').
\]
Thus
\[
(S\setminus\{s\})\cup\{s'\}\sim S\quad\text{and}\quad (S'\setminus\{s'\})\cup\{s\}\sim S'.
\]
This proves the maximizer property.

\medskip
\noindent
\textit{Coherence Property:}
Let $S,S'\subseteq\mathcal{S}$ be distinct sets with $\abs{S}=\abs{S'}=q$, and let $\hat s\in S\setminus S'$ with $(S\setminus S')\setminus\{\hat s\}\neq\emptyset$. Suppose that the three antecedents in \autoref{def:exchange-coherence} hold.

Choose any $s\in (S\setminus S')\setminus\{\hat s\}$. By \eqref{eq:rank-exchange} applied to $S$ and $S'$, there exists $s'\in S'\setminus S$ such that
\[
\rank_M(S)+\rank_M(S')\leq \rank_M\big((S\setminus\{s\})\cup\{s'\}\big)+ \rank_M\big((S'\setminus\{s'\})\cup\{s\}\big).
\]

Since $s\neq \hat s$, the set $(S\setminus\{s\})\cup\{s'\}$ contains $\hat s$ and belongs to $\nw(S\cup S')$. By condition $(ii)$, $S\succsim (S\setminus\{s\})\cup\{s'\}$, or equivalently, $\rank_M(S)\geq \rank_M((S\setminus\{s\})\cup\{s'\})$. Similarly, $(S'\setminus\{s'\})\cup\{s\}$ belongs to $\nw((S\cup S')\setminus\{\hat s\})$. By condition $(iii)$, $S'\succsim (S'\setminus\{s'\})\cup\{s\}$, or equivalently, $\rank_M(S')\geq \rank_M((S'\setminus\{s'\})\cup\{s\})$. 

Together with the rank-exchange inequality, all the above inequalities must bind, i.e., 
\[
\rank_M\big((S\setminus\{s\})\cup\{s'\}\big)=\rank_M(S) \quad\text{and}\quad \rank_M\big((S'\setminus\{s'\})\cup\{s\}\big)=\rank_M(S').
\]
Thus
\[
(S\setminus\{s\})\cup\{s'\}\sim S\quad\text{and}\quad (S'\setminus\{s'\})\cup\{s\}\sim S'.
\]
This proves the coherence property.
\end{proof}

\begin{proof}[Proof of \autoref{prop:bounds}]\

\noindent\emph{Upper-bound Property:} For any $S,S'\subseteq\mathcal{S}$ with $\abs{S}=\abs{S'}=q$, there are four possibilities. Both are feasible, in which case $S\sim S'$; both infeasible, in which case $S\sim S'$; $S$ is feasible while $S'$ is infeasible, in which case $S\succ S'$; or $S'$ is feasible while $S$ is infeasible, in which case $S'\succ S$. Since no two sets of size $q$ are incomparable, the upper-bound property holds vacuously. 

\medskip
\noindent\emph{Maximizer Property:} Let $S,S'\subseteq\mathcal{S}$ be distinct sets with $\abs{S}=\abs{S'}=q$, and suppose that $S\succsim S''$ and $S'\succsim S''$ for every $S''\in\nw(S\cup S')$.

If $S$ is infeasible, then no set in $\nw(S\cup S')$ is feasible; otherwise a feasible set would strictly dominate $S$. Hence, every set in $\nw(S\cup S')$ is infeasible, and any exchange between $S$ and $S'$ leaves both exchanged sets indifferent to the originals. Thus, the maximizer property holds.

Now suppose that $S$ is feasible. Since $S'\succsim S$, the set $S'$ is also feasible. Thus, $n_k(S)\in [L_k, U_k]$ and $n_k(S')\in [L_k, U_k]$ for all $k=1,\ldots, K$. 

Suppose $n_k(S)=n_k(S')$ for all $k$. Then there exist $s\in S\setminus S'$ and $s'\in S'\setminus S$ with $\tau(s)=\tau(s')$. Thus, $n_k(S)=n_k((S\setminus \{s\})\cup\{s'\})$ and $n_k(S')=n_k((S'\setminus \{s'\})\cup\{s\})$ for all $k$, so both exchanged sets remain feasible and hence indifferent to the originals.

Next, suppose $n_{k'}(S)\neq n_{k'}(S')$ for some $k'$. Without loss of generality, let $n_{k'}(S)>n_{k'}(S')$. Hence, there exists a student $s\in S\setminus S'$ with $\tau_{k'}(s)=1$. Moreover, there exists $k''\neq k'$ such that $n_{k''}(S)<n_{k''}(S')$; otherwise, because each student has only one protected trait, $\abs{S}=\sum_k n_k(S)>\sum_k n_k(S')=\abs{S'}$, which contradicts the fact that $\abs{S}=\abs{S'}=q$. Thus, there exists a student $s'\in S'\setminus S$ with $\tau_{k''}(s')=1$.

Let us consider the exchanged set $(S\setminus \{s\})\cup\{s'\}$. Note that 
\[
n_{k'}((S\setminus \{s\})\cup\{s'\})=n_{k'}(S)-1\geq n_{k'}(S')\geq L_{k'}
\]
where the first inequality follows because $n_{k'}(S)>n_{k'}(S')$ by assumption, and the last inequality follows because $S'$ is feasible. Similarly, 
\[
n_{k''}((S\setminus \{s\})\cup\{s'\})=n_{k''}(S)+1\leq n_{k''}(S')\leq U_{k''},
\]
where the first inequality follows because $n_{k''}(S)<n_{k''}(S')$, and the last inequality follows because $S'$ is feasible. For all other $k\notin\{k',k''\}$,  $n_{k}((S\setminus \{s\})\cup\{s'\})=n_{k}(S)$. Thus, $(S\setminus \{s\})\cup\{s'\}$ is also feasible. Consequently, $S\sim (S\setminus \{s\})\cup\{s'\}$.

An analogous argument shows that $(S'\setminus \{s'\})\cup\{s\}$ is also feasible and hence, $S'\sim (S'\setminus \{s'\})\cup\{s\}$.
Thus, the maximizer property holds.
\end{proof}

\begin{proof}[Proof of \autoref{prop:floors_ceilings}]\
The floor-and-ceiling distributional preference is an example of a distributional preference that is induced by a $q$-ordinally concave value function. Hence, \autoref{prop:floors_ceilings} follows from \autoref{prop:ordinal_concavity}. 
\end{proof}

\begin{proof}[Proof of \autoref{prop:ordinal_concavity}]\

\noindent\emph{Upper-bound Property:}
For any $S,S'\subseteq\mathcal{S}$ with $\abs{S}=\abs{S'}=q$, either $\Phi(n(S))\geq \Phi(n(S'))$, in which case $S\succsim S'$ or $\Phi(n(S'))\geq \Phi(n(S))$, in which case $S'\succsim S$. Since no two sets of size $q$ are incomparable, the upper-bound property holds vacuously. 

\medskip\noindent\emph{Maximizer Property:} Let $S,S'\subseteq\mathcal{S}$ be distinct sets with $\abs{S}=\abs{S'}=q$, and suppose that $S\succsim S''$ and $S'\succsim S''$ for every $S''\in\nw(S\cup S')$. 

Choose any $s\in (S\setminus S')$ and suppose the student has trait $t_k$. If there exist $s'\in S'\setminus S$ such that $\tau_k(s')=1$, then exchanging these students leaves trait-compositions unchanged. Hence,
\[
(S\setminus\{s\})\cup\{s'\}\sim S \quad\text{and}\quad (S'\setminus\{s'\})\cup\{s\}\sim S'.
\]

Next, suppose no such student exists in $S'\setminus S$, which implies that $n_k(S)> n_k(S')$. By $q$-ordinal concavity, there exists $k'\neq k$ with $n_{k'}(S)<n_{k'}(S')$ such that at least one of the three conditions in \autoref{def:ordinal-exchange} holds. Since $n_{k'}(S)<n_{k'}(S')$, there exists $s'\in S'\setminus S$ with trait $t_{k'}$.

However, also notice that both $S$ and $S'$ are maximal in $\nw(S\cup S')$, so we must also have $\Phi(n(S))\geq \Phi(n((S\setminus\{s\})\cup\{s'\}))$ and $\Phi(n(S'))\geq \Phi(n((S'\setminus\{s'\})\cup\{s\}))$. In other words, Condition $(i)$ and Condition $(ii)$ of \autoref{def:ordinal-exchange} cannot hold. Thus, Condition $(iii)$ of \autoref{def:ordinal-exchange} must hold, which is equivalent to 
\[
(S\setminus\{s\})\cup\{s'\}\sim S \quad\text{and}\quad (S'\setminus\{s'\})\cup\{s\}\sim S'.
\]
This proves the maximizer property.

\medskip
\noindent
\textit{Coherence Property:}
Let $S,S'\subseteq\mathcal{S}$ be distinct sets with $\abs{S}=\abs{S'}=q$, and let $\hat s\in S\setminus S'$ with $(S\setminus S')\setminus\{\hat s\}\neq\emptyset$. Suppose that the three antecedents in \autoref{def:exchange-coherence} hold. 

Choose any $s\in (S\setminus S')\setminus\{\hat s\}$ and suppose the student has trait $t_k$. If there exist $s'\in S'\setminus S$ such that $\tau_k(s')=1$, then exchanging these students leaves trait-compositions unchanged. Hence,
\[
(S\setminus\{s\})\cup\{s'\}\sim S \quad\text{and}\quad (S'\setminus\{s'\})\cup\{s\}\sim S'.
\]

Next, suppose no such student exists in $S'\setminus S$, which implies that $n_k(S)> n_k(S')$. By $q$-ordinal concavity, there exists $k'\neq k$ with $n_{k'}(S)<n_{k'}(S')$ such that at least one of the three conditions in \autoref{def:ordinal-exchange} holds. Since $n_{k'}(S)<n_{k'}(S')$, there exists $s'\in S'\setminus S$ with trait $t_{k'}$.

Since $s\neq \hat s$, the set $(S\setminus\{s\})\cup\{s'\}$ contains $\hat s$ and belongs to $\nw(S\cup S')$. By condition $(ii)$, $S\succsim (S\setminus\{s\})\cup\{s'\}$, or equivalently, $\Phi(n(S))\geq \Phi(n((S\setminus\{s\})\cup\{s'\}))$. Hence, Condition $(i)$ of \autoref{def:ordinal-exchange} is not satisfied.

Similarly, $(S'\setminus\{s'\})\cup\{s\}$ belongs to $\nw((S\cup S')\setminus\{\hat s\})$. By condition $(iii)$, $S'\succsim (S'\setminus\{s'\})\cup\{s\}$, or equivalently, $\Phi(n(S'))\geq \Phi(n((S'\setminus\{s'\})\cup\{s\}))$. Hence, Condition $(ii)$ of \autoref{def:ordinal-exchange} is also not satisfied.

Thus, Condition $(iii)$ of \autoref{def:ordinal-exchange} must hold, which is equivalent to 
\[
(S\setminus\{s\})\cup\{s'\}\sim S \quad\text{and}\quad (S'\setminus\{s'\})\cup\{s\}\sim S'.
\]
This proves the coherence property.
\end{proof}

\subsection{Proofs of Theorems}
\begin{proof}[Proof of \autoref{thm:exogenouspriority}] We first show that $(i)$ implies $(ii)$. Let the distributional preference satisfy the maximizer property. Fix a priority ranking $\priority$ and consider the distributional choice rule $\ch^{\priority}$. In \autoref{subsec:distributional}, we have already discussed that $\ch^{\priority}$ is non-wasteful, distributionally maximal, and free of justified envy. Thus, we need only show that $\ch^{\priority}$ is the unique choice rule that satisfies these three axioms.

To that end, let $\ch$ be any choice rule that is non-wasteful, distributionally maximal, and free of justified envy. By \autoref{lem:distfrontier}, $\ch(S)\in\front(S)$ for every $S\subseteq\mathcal{S}$. Suppose, toward a contradiction, that there exists $S\subseteq\mathcal{S}$ such that $\ch(S)\neq \ch^{\priority}(S)$. Such a set $S$ is necessarily non-empty. 

Since the distributional preference satisfies the upper-bound and maximizer properties, \autoref{lem:matroidbase} implies that $\front(S)$ forms the bases of a matroid. By the weak simultaneous exchange property (see \ref{sec:matroid_theory}), there exist $s\in \ch(S)\setminus\ch^{\priority}(S)$ and $s'\in \ch^{\priority}(S)\setminus\ch(S)$ such that
\[
(\ch(S)\setminus\{s\})\cup\{s'\}\in\front(S) \quad \text{ and } \quad (\ch^{\priority}(S)\setminus\{s'\})\cup\{s\}\in\front(S).
\]
By \autoref{prop:comp}, $(\ch(S)\setminus\{s\})\cup\{s'\}\sim \ch(S)$ and $(\ch^\priority(S)\setminus\{s'\})\cup\{s\}\sim \ch^\priority(S)$. 

Since $\ch$ is free of justified envy, $(\ch(S)\setminus\{s\})\cup\{s'\}\sim \ch(S)$ implies that $s\mathrel{\priority}s'$. At the same time, since $\ch^\priority$ is also free of justified envy, $(\ch^\priority(S)\setminus\{s'\})\cup\{s\}\sim \ch^\priority(S)$ implies $s'\mathrel{\priority}s$. However, as the priority ranking $\priority$ is a strict linear order on $\mathcal{S}$, we cannot have both $s\mathrel{\priority}s'$ and $s'\mathrel{\priority}s$. Therefore, no such set $S$ exists, and $\ch=\ch^\priority$. 
\medskip

We now show that $(ii)$ implies $(i)$. For each priority ranking $\priority$, let $\ch^\priority$ be the unique choice rule that is non-wasteful, distributionally maximal, and free of justified envy. Toward a contradiction, suppose that the maximizer property does not hold. Since the upper-bound property is satisfied, \autoref{lem:matroidbase} implies that there exists a set $T\subseteq\mathcal S$ such that $\front(T)$ does not form the bases of a matroid with ground set $T$.

Since $\front(T)$ is non-empty, it must violate the matroid base-exchange axiom (B2), i.e., there exist distinct sets $S,S'\in\front(T)$ and some $i\in S\setminus S'$ such that for every $s'\in S'\setminus S$, we have
\[
(S\setminus\{i\})\cup\{s'\}\notin\front(T).
\]
Among all tuples $(T,S,S', i)$ that violate (B2), choose the tuple for which $\abs{S\setminus S'}$ is minimal.

Since $S\in\front(T)$ and $(S\setminus\{i\})\cup\{s'\}\in\nw(T)\setminus \front(T)$ for every $s'\in S'\setminus S$, \autoref{prop:comp} implies
\begin{equation}\label{eq:maxstar-failure}
S\succ (S\setminus\{i\})\cup\{s'\}
\qquad\text{for every }s'\in S'\setminus S.
\end{equation}
Moreover, $\abs{S\setminus S'}\ge 2$. Indeed, if $S\setminus S'=\{i\}$, then exchanging $i$ with the unique student  $s'\in S'\setminus S$ yields $(S\setminus\{i\})\cup\{s'\}=S'\sim S$, which contradicts \eqref{eq:maxstar-failure}.

\begin{claim}\label{claim:1}
For every $j\in (S\setminus S')\setminus\{i\}$ and every $k\in S'\setminus S$,
\begin{equation}\label{eq:key-claim-theorem1}
S'\succ (S'\setminus\{k\})\cup\{j\}.
\end{equation}\end{claim}
\begin{proof}
Toward a contradiction, suppose there exists some such $j$ and $k$ such that $S'\not\succ (S'\setminus\{k\})\cup\{j\}$. Let $\widetilde S':=(S'\setminus\{k\})\cup\{j\}$. Since $S'\in\front(T)$ and $\widetilde S'\in\nw(T)$, by \autoref{prop:comp}, the fact that $S'\not\succ\widetilde S'$ implies that $\widetilde S'\in\front(T)$.  

Moreover, $i\in S\setminus\widetilde S'$, and for every $\tilde s\in \widetilde S'\setminus S=(S'\setminus S)\setminus\{k\}$, \eqref{eq:maxstar-failure} gives
\[
S\succ (S\setminus\{i\})\cup\{\tilde s\},
\]
which implies that $(S\setminus\{i\})\cup\{\tilde s\}\notin \front(T)$. Thus, the tuple $(T,S,\widetilde S',i)$ also violates the base-exchange condition. However, $\abs{S\setminus\widetilde S'}=\abs{S\setminus S'}-1$, contradicting the fact that the tuple $(T,S,S',i)$ was chosen to minimize $\abs{S\setminus S'}$. Therefore \eqref{eq:key-claim-theorem1} holds.
\end{proof}

We now construct a priority ranking $\pi$ as follows: Rank students in $S\cap S'$ first, then students in $(S\setminus S')\setminus\{i\}$, then students in $S'\setminus S$, then $i$, and finally rank all students outside $S\cup S'$ last. Within each block, the students are ranked in arbitrary order. Thus, for every $\tilde s\in S\cap S'$, $s\in (S\setminus S')\setminus\{i\}$, $s'\in S'\setminus S$, and $t\in T\setminus (S\cup S')$, we have 
\[
\tilde s\mathrel{\pi}s\mathrel{\pi}s'\mathrel{\pi}i \mathrel{\pi} t.
\]

We first verify that $\ch^\priority(T)=S$. The greedy algorithm admits every student in $S\cap S'$, since this set is contained in the frontier set $S$. It then admits every student in $(S\setminus S')\setminus\{i\}$, since the tentative set remains contained in $S$. At that point, the tentative set is $S\setminus\{i\}$. For any $s'\in S'\setminus S$, adding $s'$ to the tentative set would produce $(S\setminus\{i\})\cup\{s'\}$, which is not in $\front(T)$ by \eqref{eq:maxstar-failure}. Hence, every student in $S'\setminus S$ is rejected. Finally, $i$ is considered and admitted because $(S\setminus\{i\})\cup\{i\}=S\in\front(T)$. Once $i$ is admitted, the tentative set is $S$, which has cardinality $q$; hence no later student can be admitted. Thus $\ch^\priority(T)=S$.

Define a new choice rule $\ch$ by
\[
\ch(T')=\begin{cases}
S' & \text{if } T'=T,\\
\ch^\priority(T') & \text{otherwise}.
\end{cases}
\]
Because $\ch(T)=S'$ whereas $\ch^\priority(T)=S$, $\ch\neq \ch^\priority$. 

Notice that for every input $T'\subseteq \mathcal{S}$, $\ch(T')\in\front(T')$. Indeed, for input $T'\neq T$, $\ch(T')=\ch^\priority(T')$, and $\ch^\priority(T')\in\front(T')$ by construction. For input $T'=T$, $\ch(T)=S'\in\front(T)$. Therefore, by \autoref{lem:distfrontier}, $\ch$ is non-wasteful and distributionally maximal. 

It remains to verify that $\ch$ is free of justified envy. For every input $T'\neq T$, this follows because $\ch$ coincides with $\ch^\priority$, which is free of justified envy. Hence, it suffices to check the input $T'=T$.

Let $k\in \ch(T)=S'$ and $j\in T\setminus S'$ be such that $(S'\setminus\{k\})\cup\{j\}\succsim S'$ so that the antecedent of the free-of-justified-envy axiom holds. We show that $k\mathrel{\pi}j$. There are four cases.

\medskip
\noindent
\textit{Case 1: $j=i$.} Since $i$ is ranked below every student in $S'$, we have $k\mathrel{\pi}j$.

\medskip
\noindent
\textit{Case 2: $j\in (S\setminus S')\setminus\{i\}$ and $k\in S\cap S'$.} By construction of $\pi$, every student in $S\cap S'$ is ranked above every student in $(S\setminus S')\setminus\{i\}$. Hence $k\mathrel{\pi}j$.

\medskip
\noindent
\textit{Case 3: $j\in (S\setminus S')\setminus\{i\}$ and $k\in S'\setminus S$.} By \eqref{eq:key-claim-theorem1}, $S'\succ (S'\setminus\{k\})\cup\{j\}$, which contradicts the maintained hypothesis $(S'\setminus\{k\})\cup\{j\}\succsim S'$. Thus this case cannot occur.

\medskip
\noindent
\textit{Case 4: $j\in T\setminus(S\cup S')$}. By construction of $\pi$, every student in $S'$ is ranked above every student in $T\setminus (S\cup S')$. Hence $k\mathrel{\pi}j$.

In every case in which the antecedent of the free-of-justified-envy axiom is satisfied, the selected student $k$ has higher priority than the rejected student $j$. Therefore $\ch$ is free of justified envy.

We have therefore constructed a choice rule $\ch\neq \ch^\priority$ that is non-wasteful, distributionally maximal, and free of justified envy. This contradicts the initial assertion that $\ch^\priority$ is the unique such choice rule for every priority ranking. To avoid such a contradiction, the maximizer property must hold.
\end{proof}

\begin{proof}[Proof of \autoref{thm:meritgeneral}]
Let $\ch$, $\priority$, and $S$ be as in the statement of \autoref{thm:meritgeneral}. Suppose that condition $(i)$ does not hold, i.e., $\abs{\ch^{\priority}(S)}\leq \abs{\ch(S)}$. We show that either condition $(ii)$ or condition $(iii)$ holds.

Because $\ch(S)\subseteq S$, we have $\abs{\ch(S)}\leq \abs{S}$. Additionally, all choice rules must respect the capacity, so $\abs{\ch(S)}\leq q$. Hence,  $\abs{\ch(S)}\leq \min\{q,\abs{S}\}$. Since $\ch^{\priority}$ is non-wasteful, we have $\min\{q,\abs{S}\}=\abs{\ch^\priority(S)}\leq \abs{\ch(S)}$, implying that $\abs{\ch(S)}= \min\{q,\abs{S}\}$.

If $\abs{S}\le q$, then $\abs{\ch^\priority(S)}=\abs{\ch(S)}=\abs{S}$, and hence $\ch^{\priority}(S)=\ch(S)=S$, contradicting the premise that $\ch(S)\neq\ch^{\priority}(S)$. Therefore, $\abs{S}>q$. We thus have $\abs{\ch^{\priority}(S)}=\abs{\ch(S)}=q$, implying that $\ch(S)\in\nw(S)$.

If $\ch(S)\notin\front(S)$, then \autoref{prop:comp} implies $\ch^{\priority}(S)\succ \ch(S)$, because $\ch^{\priority}(S)\in\front(S)$ and $\ch(S)\in\nw(S)\setminus\front(S)$. Hence condition $(ii)$ holds.

Now suppose $\ch(S)\in\front(S)$. By \autoref{lem:matroidbase}, $\front(S)$ forms the bases of a matroid. In this matroid, a set is independent if and only if it is contained in some frontier set, which is exactly the algorithm used by $\ch^\priority$ at each step. Hence $\ch^{\priority}(S)$ is the greedy basis of this matroid under $\priority$. Since $\ch(S)\in\front(S)$ is also a basis, and hence an independent set, \autoref{lem:Gale} implies that $\ch^{\priority}(S)$ priority dominates $\ch(S)$. Hence condition $(iii)$ holds.

Thus, whenever condition $(i)$ fails, either condition $(ii)$ or condition $(iii)$ holds. Therefore at least one of the three conditions in \autoref{thm:meritgeneral} must hold.
\end{proof}

%%%%%%%%%%%%%%%%%%%%%%%%%%%%%%%%%%%%%%%%%%%%%%%%%%%%%%%%%%%%%%%%%%%%%%%%%%%%%
\begin{proof}[Proof of \autoref{thm:genpi}]
We first show that $(i)$ implies $(ii)$. Suppose the distributional preference satisfies the maximizer and coherence properties. We show that, for every priority ranking $\priority$, the distributional choice rule $\ch^{\priority}$ satisfies path independence. We use the following characterization of path independence due to \citet{aizmal81}: a choice rule $\ch$ is path independent if and only if, for every $S\subseteq\mathcal{S}$ and every $s\in S$, 
\begin{enumerate}[label={$(\alph*)$}, parsep=0pt, itemsep=2pt]
    \item \textbf{Irrelevance of rejected students (IRS):} If $s\notin \ch(S)$, then $\ch(S\setminus\{s\})=\ch(S)$.
    \item \textbf{Substitutability:} If $s\in \ch(S)$, then for any $\hat s\in S\setminus\{s\}$, $s\in\ch(S\setminus\{\hat s\})$.
\end{enumerate} 
We shall show that the distributional choice rule $\ch^\priority$ satisfies IRS and substitutability, and hence, is path independent.

\medskip
\noindent\textit{IRS.} Consider an applicant set $S\subseteq\mathcal{S}$, and suppose $s\in S\setminus \ch^{\priority}(S)$. Since $\ch^\priority$ is non-wasteful, this
implies $\abs{S}>q$. Hence both $S$ and $S\setminus\{s\}$ have non-wasteful sets
of cardinality $q$.

By \autoref{lem:distfrontier}, $\ch^{\priority}(S)\in\front(S)$. Since $s\notin\ch^\priority(S)$, the set $\ch^\priority(S)$ belongs to $\nw(S\setminus\{s\})$. Moreover, because $\ch^{\priority}(S)$ is maximal in $\nw(S)$, it remains maximal in the smaller set $\nw(S\setminus\{s\})$, implying $\ch^{\priority}(S)\in\front(S\setminus\{s\})$. 

Next, consider the applicant set $S\setminus\{s\}$. By \autoref{lem:distfrontier}, $\ch^{\priority}(S\setminus\{s\})\in \front(S\setminus\{s\})$. Since both $\ch^{\priority}(S\setminus\{s\})$ and $\ch^{\priority}(S)$ belong to $\front(S\setminus\{s\})$, \autoref{prop:comp} implies $\ch^{\priority}(S\setminus\{s\})\sim \ch^{\priority}(S)$. Consequently, $\ch^{\priority}(S\setminus\{s\})\in\front(S)$; otherwise, $\ch^{\priority}(S)$ would strictly dominate $\ch^{\priority}(S\setminus\{s\})$ by \autoref{prop:comp}, contradicting the indifference.

By \autoref{lem:matroidbase}, $\front(S\setminus\{s\})$ forms the bases of a matroid on $S\setminus\{s\}$. Under this matroid, $\ch^\priority(S\setminus\{s\})$ is the greedy basis, while $\ch^\priority(S)$ is also a basis. Applying \autoref{lem:Gale}, $\ch^\priority(S\setminus\{s\})$ priority dominates $\ch^\priority(S)$.

Similarly, $\front(S)$ forms the bases of a matroid on $S$. Under this matroid, $\ch^\priority(S)$ is the greedy basis, while $\ch^\priority(S\setminus\{s\})$ is also a basis. Applying \autoref{lem:Gale} again, $\ch^\priority(S)$ priority dominates $\ch^\priority(S\setminus\{s\})$.

Since the two sets priority dominate one another and $\priority$ is a strict linear order, they must be equal. Hence $\ch^{\priority}$ satisfies IRS.

\medskip
\noindent\textit{Substitutability.}
Let $S\subseteq\mathcal{S}$, let $s\in \ch^{\priority}(S)$, and let $\hat s\in S\setminus\{s\}$. If $\hat s\notin \ch^{\priority}(S)$, then IRS gives $\ch^{\priority}(S\setminus\{\hat s\})=\ch^{\priority}(S)$, and the claim follows. Hence, suppose that $\hat s\in \ch^{\priority}(S)$. If $\abs{S\setminus\{\hat s\}}\leq q$, then non-wastefulness gives $\ch^{\priority}(S\setminus\{\hat s\})=S\setminus\{\hat s\}$, so the claim again follows trivially. Thus, also assume that $\abs{S\setminus\{\hat s\}}>q$.

Set $A:=\ch^{\priority}(S)$ and $B:=\ch^{\priority}(S\setminus\{\hat s\})$. Suppose, toward a contradiction, that $s\notin B$. We shall consider two cases and show each case leads to a contradiction. Since $\abs{S\setminus \{\hat s\}}>q$ and $\ch^\priority$ is non-wasteful, we have $\abs{A}=\abs{B}=q$. Additionally, $s,\hat s\in A\setminus B$ and $s\in (A\setminus B)\setminus\{\hat s\}$, so $(A\setminus B)\setminus\{\hat s\}\neq \varnothing$.

\medskip
\noindent\textit{Case 1: $B\notin\front(S)$.}
Since $A\in\front(S)$ and $B\notin\front(S)$, \autoref{prop:comp} gives $A\succ B$. Recall that $\abs{S\setminus \{\hat s\}}>q$ and $(A\setminus B)\setminus\{\hat s\}\neq \varnothing$. Thus, the domain restrictions for the coherence property are satisfied. We verify the three antecedents of the coherence property for the pair $A,B$ and the student $\hat s$.

Take any $T\in\nw(A\cup B)$ with $\hat s\notin T$. Since $\abs{A}=\abs{B}=q$ and $(A\cup B)\setminus\{\hat s\}\subseteq S\setminus\{\hat s\}$, we have $T\in\nw(S\setminus\{\hat s\})$. Since $B\in\front(S\setminus\{\hat s\})$, \autoref{prop:comp} gives $B\succsim T$. Combined with $A\succ B$ and transitivity, we conclude that $A\succ T$ for all $T\in\nw(A\cup B)$ with $\hat s\notin T$. Hence, the first antecedent holds.

Since $A$ is maximal in $\nw(S)$ and $A\cup B\subseteq S$, $A$ remains maximal in the smaller set $\nw(A\cup B)$, i.e., $A\in\front(A\cup B)$. By \autoref{prop:comp}, $A\succsim T$ for all $T\in\nw(A\cup B)$, so the second antecedent holds.

Similarly, since $B$ is maximal in $\nw(S\setminus\{\hat s\})$ and $(A\cup B)\setminus\{\hat s\}\subseteq S\setminus\{\hat s\}$, $B$ remains maximal in the smaller set $\nw((A\cup B)\setminus\{\hat s\})$, i.e., $B\in\front((A\cup B)\setminus\{\hat s\})$. By \autoref{prop:comp}, $B\succsim T$ for all $T\in\nw((A\cup B)\setminus\{\hat s\})$, so the third antecedent holds.

We can thus apply the coherence property: there exists $t\in (A\setminus B)\setminus\{\hat s\}$ and $t'\in B\setminus A$ such that
\[
(A\setminus\{t\})\cup\{t'\}\sim A \quad\text{and}\quad (B\setminus\{t'\})\cup\{t\}\sim B.
\]
Since $A\in\front(S)$ and $(A\setminus\{t\})\cup\{t'\}\in \nw(S)$, the first indifference and \autoref{prop:comp} imply $(A\setminus\{t\})\cup\{t'\}\in\front(S)$. Additionally, notice that $(B\setminus\{t'\})\cup\{t\}\in \nw(S\setminus\{\hat s\})$, because $B\setminus\{t'\}\subseteq B\subseteq S\setminus\{\hat s\}$ and $t\neq \hat s$. Since $B\in\front(S\setminus\{\hat s\})$, the second indifference and \autoref{prop:comp} imply $(B\setminus\{t'\})\cup\{t\}\in\front(S\setminus\{\hat s\})$.

By \autoref{lem:matroidbase}, $\front(S\setminus\{\hat s\})$ forms the bases of a matroid on $S\setminus\{\hat s\}$. Under this matroid, $B$ is the greedy basis, while $(B\setminus\{t'\})\cup\{t\}$ is also a basis. Applying \autoref{lem:Gale}, $B$ priority dominates $(B\setminus\{t'\})\cup\{t\}$. Because these two sets differ only by replacing $t'$ with $t$, this priority dominance implies $t'\mathrel{\priority} t$.

Similarly, by \autoref{lem:matroidbase}, $\front(S)$ forms the bases of a matroid on $S$. Under this matroid, $A$ is the greedy basis, while $(A\setminus\{t\})\cup\{t'\}$ is also a basis. Applying \autoref{lem:Gale}, $A$ priority dominates $(A\setminus\{t\})\cup\{t'\}$. Because these two sets differ only by replacing $t$ with $t'$, this priority dominance implies $t\mathrel{\priority} t'$. However, this contradicts the asymmetry of the strict linear order $\priority$. 

\medskip
\noindent\textit{Case 2: $B\in\front(S)$.} By \autoref{lem:matroidbase}, $\front(S)$ forms the bases of a matroid on $S$. Since $A$ and $B$ are both in $\front(S)$ and since $s\in A\setminus B$, the strong basis exchange property gives some $s'\in B\setminus A$ such that
\[
(A\setminus\{s\})\cup\{s'\}\in\front(S) \quad\text{and}\quad (B\setminus\{s'\})\cup\{s\}\in\front(S).
\]

Since $B\setminus\{s'\}\subseteq B\subseteq S\setminus\{\hat s\}$ and $s\in(A\setminus B)\setminus\{\hat s\}$, we have that $(B\setminus\{s'\})\cup\{s\}\subseteq S\setminus \{\hat s\}$. Additionally, since $\abs{B}=\abs{(B\setminus\{s'\})\cup\{s\}}=q$, we conclude that $(B\setminus\{s'\})\cup\{s\}\in \nw(S\setminus\{\hat s \})$. In fact, since $(B\setminus\{s'\})\cup\{s\}\in\front(S)$, it remains maximal in $\nw(S\setminus\{\hat s\})$. Thus, $(B\setminus\{s'\})\cup\{s\}\in\front(S\setminus\{\hat s\})$.

By \autoref{lem:matroidbase}, $\front(S\setminus\{\hat s\})$ forms the bases of a matroid on $S\setminus\{\hat s\}$. Under this matroid, $B$ is the greedy basis, while $(B\setminus\{s'\})\cup\{s\}$ is also a basis. Applying \autoref{lem:Gale}, $B$ priority dominates $(B\setminus\{s'\})\cup\{s\}$. Because these two sets differ only by replacing $s'$ with $s$, this priority dominance implies $s'\mathrel{\priority} s$. 

On the other hand, under the matroid whose bases are $\front(S)$, $A$ is the greedy basis, while $(A\setminus\{s\})\cup\{s'\}$ is also a basis. Applying \autoref{lem:Gale}, $A$ priority dominates $(A\setminus\{s\})\cup\{s'\}$. Because these two sets differ only by replacing $s$ with $s'$, this priority dominance implies $s\mathrel{\priority} s'$. However, this contradicts the asymmetry of the strict linear order $\priority$. Again we obtain a contradiction.

The two cases exhaust all possibilities, so the supposition that $s\notin\ch^\priority(S\setminus\{\hat{s}\})$ is false. As a result, $\ch^\priority$ satisfies substitutability. Since $\ch^\priority$ also satisfies IRS, the characterization of \citet{aizmal81} implies that $\ch^\priority$ satisfies path independence. 
\bigskip

We now prove that $(ii)$ implies $(i)$. Suppose that, for every priority ranking $\priority$, the distributional choice rule $\ch^{\priority}$ satisfies path independence. We first show that the distributional preference satisfies the maximizer property. Then, taking the maximizer property as a given, we show that the distributional preference also satisfies the coherence property.

\noindent\textit{Path independence$\implies$maximizer.} Toward a contradiction, suppose that the maximizer property does not hold. Since the upper-bound property is satisfied, \autoref{lem:matroidbase} implies that there exists a set $T\subseteq\mathcal S$ such that $\front(T)$ does not form the bases of a matroid with ground set $T$.

Since $\front(T)$ is non-empty, it must violate the matroid base-exchange axiom (B2), i.e., there exist distinct sets $S,S'\in\front(T)$ and some $s\in S\setminus S'$ such that for every $s'\in S'\setminus S$, we have
\[
(S\setminus\{s\})\cup\{s'\}\notin\front(T).
\]

Choose a priority ranking $\priority$ that ranks all students in $S\setminus\{s\}$ first, then all students in $S'\setminus S$, then $s$, and finally all students in $T\setminus(S\cup S')$, with an arbitrary ranking across students within each block. Under this priority ranking, the distributional choice rule applied to $T$ selects $\ch^\priority(T)=S$. Indeed, every student in $S\setminus\{s\}$ is admitted because the set selected so far remains contained in the frontier set $S$. Each student $s'\in S'\setminus S$ is rejected because $(S\setminus\{s\})\cup\{s'\}$ has cardinality $q$ and is not in $\front(T)$. Finally, $s$ is admitted because $S\in\front(T)$. 

Since $\ch^\priority$ is path independent, it satisfies IRS and substitutability. Because no student in $T\setminus(S\cup S')$ belongs to $\ch^\priority(T)=S$, repeated application of IRS gives $\ch^\priority(T)=\ch^\priority(S\cup S')=S$. Now remove $s$ from $S\cup S'$. By substitutability, every student in $S\setminus\{s\}$ must remain chosen, so $S\setminus\{s\}\subseteq \ch^\priority((S\cup S')\setminus\{s\})$.

Because $S'\subseteq (S\cup S')\setminus\{s\}$, the reduced applicant pool has at least $q$ students, so $\abs{\ch^\priority((S\cup S')\setminus\{s\})}=q$. Hence, there exists some $s^*\in S'\setminus S$ such that
\[
\ch^\priority((S\cup S')\setminus\{s\})=(S\setminus\{s\})\cup\{s^*\}.
\]
By construction, the distributional choice rule always picks a set in the frontier. Thus, we must have $(S\setminus\{s\})\cup\{s^*\}\in\front((S\cup S')\setminus\{s\})$.

Additionally, notice that $S'\in \front((S\cup S')\setminus\{s\})$; because $S'$ is maximal in $T$, it remains so in the reduced pool $(S\cup S')\setminus\{s\}$. By \autoref{prop:comp}, $S'\sim (S\setminus\{s\})\cup\{s^*\}$. Using the fact that $S'\in \front(T)$ and $(S\setminus\{s\})\cup\{s^*\}\in \nw(T)$, we once again apply \autoref{prop:comp} to get that $(S\setminus\{s\})\cup\{s^*\}\in\front(T)$. However, this contradicts the assumption that, for every $s'\in S'\setminus S$,
\[
(S\setminus\{s\})\cup\{s'\}\notin\front(T).
\]
Hence, to avoid the contradiction, the maximizer property must hold.

\medskip
\noindent\textit{Path independence + maximizer$\implies$coherence.} Take $S,S'\subseteq\mathcal{S}$ with $\abs{S}=\abs{S'}=q$, and take $\hat s\in S\setminus S'$ satisfying $(S\setminus S')\setminus\{\hat s\}\neq\varnothing$. Suppose $S, S'$ and $\hat s$ satisfy the three antecedents in \autoref{def:exchange-coherence}:
\begin{align}
    &S\succ T \text{ for every }T\in\nw(S\cup S') \text{ with } \hat s\notin T,\label{eq:ec-i} \\
    &S\succsim T \text{ for every }T\in\nw(S\cup S') \text{ with } \hat s\in T,\label{eq:ec-ii}\\
    &S'\succsim T \text{ for every }T\in\nw\bigl((S\cup S')\setminus\{\hat s\}\bigr).\label{eq:ec-iii}
\end{align}

Let $U:=S\cup S'$ and $W:=U\setminus\{\hat s\}$. Clearly, $q\leq\abs{W}\leq \abs{U}$. By \eqref{eq:ec-i} and \eqref{eq:ec-ii}, $S\in\front(U)$: for every $T\in\nw(U)$, either $\hat s\notin T$ and $S\succ T$ by \eqref{eq:ec-i}, or $\hat s\in T$ and $S\succsim T$ by \eqref{eq:ec-ii}. Thus, no element of $\nw(U)$ strictly dominates $S$. Moreover, every set in $\front(U)$ contains $\hat s$: if $T\in\nw(U)$ and $\hat s\notin T$, then $S\succ T$ by \eqref{eq:ec-i}, so $T\notin\front(U)$. Finally, \eqref{eq:ec-iii} implies $S'\in\front(W)$, because $S'\in\nw(W)$ and $S'$ weakly dominates every element of $\nw(W)$.

By \autoref{lem:matroidbase}, $\front(U)$ forms the bases of a matroid on $U$. Denote this matroid by $M$, and note that it has rank $q$. Since every basis of $M$ contains $\hat{s}$, contracting $M$ by $\hat{s}$ yields a matroid on $W$ whose bases are
\[
    \mathcal{B}^c\coloneqq \left\{T\setminus\{\hat s\}:T\in\front(U)\right\}.
\]
Denote this matroid by $M^c$, and note that $M^c$ has rank $q-1$. 

Separately, by \autoref{lem:matroidbase}, $\front(W)$ forms the bases of a matroid on $W$. Denote this matroid by $M^d$, and note that $M^d$ has rank $q$.

Fix an arbitrary priority ranking $\rho$ on $W$. Extend $\rho$ to a priority ranking $\priority$ on $U$ by ranking $\hat s$ above every element of $W$ and ordering the elements of $W$ according to $\rho$. Let $G^\rho_{M^c}$ and $G^\rho_{M^d}$ denote the greedy basis under $\rho$ of matroids $M^c$ and $M^d$, respectively.

Because every set in $\front(U)$ contains $\hat s$, and because $\hat s$ has the highest priority under $\priority$, $\ch^\priority$ applied to $U$ first admits $\hat s$. Then, the choice rule proceeds along $W$ in $\rho$-order, admitting at each step the highest-$\rho$ element compatible with extending to a set in $\front(U)$. Thus, once $\hat s$ has been admitted, the choice rule is precisely the greedy algorithm on $M^c$ under $\rho$, so
\begin{equation}\label{eq:thm3-greedy-minus}
    \ch^{\priority}(U)=\{\hat s\}\cup G^\rho_{M^c}.
\end{equation}

Notice that $\ch^\priority=\ch^\rho$ on $W$ since $\rho$ and $\priority$ coincide on $W$. When applying $\ch^\priority$ to $W$, the choice rule proceeds along $W$ in $\rho$-order, admitting at each step the highest-$\rho$ element compatible with extending to a set in $\front(W)$. Hence,
\begin{equation}\label{eq:thm3-greedy-plus}
    \ch^{\priority}(W)=G^\rho_{M^d}.
\end{equation}

By assumption, $\ch^{\priority}$ satisfies path independence. Thus, 
\[
\ch^\priority(U)=\ch^\priority(W\cup \{\hat s\})=\ch^\priority(\ch^\priority(W)\cup \{\hat s\})
\]
By \eqref{eq:thm3-greedy-minus} and \eqref{eq:thm3-greedy-plus}, we get
\[
\{\hat s\}\cup G^\rho_{M^c}=\ch^\priority(G^\rho_{M^d}\cup \{\hat s\})\subseteq \{\hat s\}\cup G^\rho_{M^d},
\]
where the set inclusion follows by definition of a choice rule. Therefore, $G^\rho_{M^c}\subseteq G^\rho_{M^d}$. Since $\rho$ was arbitrary, we have that $G^\rho_{M^c}\subseteq G^\rho_{M^d}$ for every priority ranking $\rho$ on $W$. By \autoref{lem:nested-greedy-quotient}, we can conclude that $M^c$ is a quotient of $M^d$. Moreover, recall that $M^c$ has rank $q-1$ and $M^d$ has rank $q$. Thus, $M^c$ is a rank-one quotient of $M^d$. 

Now observe that $S\setminus\{\hat s\}\in\mathcal{B}^c$, $S'\in \front(W)$, and $(S\setminus S')\setminus\{\hat s\}\neq
\varnothing$ by assumption. Applying \autoref{lem:nested-greedy-exchange}, there exist
$s\in (S\setminus \{\hat s\})\setminus S'$ and $s'\in S'\setminus S$ such that
\[
    ((S\setminus \{\hat s\})\setminus\{s\})\cup\{s'\}\in\mathcal{B}^c\quad\text{and}\quad (S'\setminus \{s'\})\cup\{s\}\in\front(W).
\]

By the definition of $\mathcal{B}^c$, adding $\hat s$ to the first set yields $(S\setminus\{s\})\cup\{s'\}\in\front(U)$. Since $S\in\front(U)$, \autoref{prop:comp} implies
$(S\setminus\{s\})\cup\{s'\}\sim S$.
Since $S'\in\front(W)$ and $(S'\setminus\{s'\})\cup\{s\}\in\front(W)$, \autoref{prop:comp} also implies
$(S'\setminus\{s'\})\cup\{s\}\sim S'$.
Thus the conclusion of the coherence property holds. Therefore, the distributional preference satisfies the coherence property.
\end{proof}

%%%%%%%%%%%%%%%%%%%%%%%%%%%%%%%%%%%%%%%%%%%%%%%%%%%%%%%%%%%%%%%%%%%%%%%%%%%%%%

\begin{proof}[Proof of \autoref{thm:revealed_priority}]\

\noindent\textbf{($\Leftarrow$).}
Consider $\ch^{\priority}$ for some priority ranking $\priority$. By construction, $\ch^{\priority}(S)\in\front(S)$ for every $S\subseteq\mathcal{S}$. Therefore, by \autoref{lem:distfrontier}, $\ch^{\priority}$ is non-wasteful and distributionally maximal.

We show that $\ch^{\priority}$ satisfies the revealed priority axiom. Let $S\subseteq\mathcal{S}$, $s\in\ch^{\priority}(S)$, and $s'\in S\setminus\ch^{\priority}(S)$, and suppose that
\[
\bigl(\ch^{\priority}(S)\setminus\{s\}\bigr)\cup\{s'\}\succsim\ch^{\priority}(S).
\]
Since $\ch^{\priority}(S)\in\front(S)$ and the swapped set belongs to $\nw(S)$, \autoref{prop:comp} implies that $\bigl(\ch^{\priority}(S)\setminus\{s\}\bigr)\cup\{s'\}\in\front(S)$. By \autoref{lem:matroidbase}, $\front(S)$ forms the bases of a matroid. Since $\ch^{\priority}(S)$ is the greedy basis of this matroid under $\priority$, \autoref{lem:Gale} implies that $\ch^{\priority}(S)$ priority dominates $\bigl(\ch^{\priority}(S)\setminus\{s\}\bigr)\cup\{s'\}$. These two sets differ only by replacing $s$ with $s'$, so priority domination implies $s\mathrel{\priority}s'$.

Now suppose, toward a contradiction, that there exist sequences $\{S_k\}_{k=1}^K$ and $\{s_k\}_{k=1}^K$ satisfying the three conditions of the revealed priority axiom. The preceding argument implies $s_k\mathrel{\priority}s_{k+1}$ for every $k\in\{1,\ldots, K\}$. Hence, 
\[
s_1\mathrel{\priority}s_2\mathrel{\priority}\cdots\mathrel{\priority}s_K\mathrel{\priority}s_1.
\]
By transitivity, this implies $s_1\mathrel{\priority}s_1$, contradicting the asymmetry of $\priority$. Thus, $\ch^{\priority}$ satisfies the revealed priority axiom.

\medskip
\noindent
\noindent\textbf{($\Rightarrow$).}
Suppose that $\ch$ is non-wasteful, distributionally maximal, and satisfies the revealed priority axiom. By \autoref{lem:distfrontier}, $\ch(S)\in\front(S)$ for every $S\subseteq\mathcal{S}$.

Define a binary relation $\rho$ on $\mathcal{S}$ as follows: $s\mathrel{\rho}s'$ if there exists $S\subseteq\mathcal{S}$ such that $s\in\ch(S)$, $s'\in S\setminus\ch(S)$, and $(\ch(S)\setminus\{s\})\cup\{s'\}\succsim\ch(S)$. The revealed priority axiom implies that $\rho$ is acyclic. Therefore, by a generalization of Szpilrajn's lemma \citep[Theorem~1.5]{chambers2016revealed}, $\rho$ can be completed to a priority ranking over $\mathcal{S}$. Let $\priority$ be such a completion, and let $\ch^{\priority}$ be the associated distributional choice rule. By construction, $\ch^\priority(S)\in\front(S)$ for every $S\subseteq\mathcal{S}$.

Toward a contradiction, suppose that there exists $S\subseteq\mathcal{S}$ such that $\ch(S)\neq\ch^\priority(S)$. Since the distributional preference satisfies the upper-bound and maximizer properties, \autoref{lem:matroidbase} implies that $\front(S)$ forms the bases of a matroid. By the strong basis exchange property, there exist $s\in\ch(S)\setminus\ch^\priority(S)$ and $s'\in\ch^\priority(S)\setminus\ch(S)$ such that 
\[
(\ch(S)\setminus\{s\})\cup\{s'\}\in\front(S)\quad \text{ and }\quad (\ch^\priority(S)\setminus\{s'\})\cup\{s\}\in\front(S).
\]

By \autoref{prop:comp}, $(\ch(S)\setminus\{s\})\cup\{s'\}\sim\ch(S)$. Therefore, by the definition of $\rho$, we have $s\mathrel{\rho} s'$. Since $\priority$ extends $\rho$, it follows that $s\mathrel{\priority}s'$.

On the other hand, since $(\ch^\priority(S)\setminus\{s'\})\cup\{s\}\in\front(S)$, \autoref{lem:Gale} implies that $\ch^\priority(S)$ priority dominates $(\ch^\priority(S)\setminus\{s'\})\cup\{s\}$ under $\priority$. These two sets differ only by replacing $s'$ with $s$, so priority domination implies $s'\mathrel{\priority}s$. However, this contradicts our previous conclusion that $s\mathrel{\priority}s'$. Hence no such $S$ exists, and therefore $\ch=\ch^{\priority}$ for some priority ranking $\priority$.
\end{proof}

\begin{proof}[Proof of \autoref{thm:da}]
We use the deferred-acceptance characterization of \citet{Dogan/Imamura/Yenmez:2025}. Their result applies to profiles of choice rules that are characterized by \textit{solitary} axioms and that also satisfy path independence and \textit{size monotonicity}. We define these terms shortly. They show that the corresponding mechanism-level extensions of those solitary axioms, together with individual rationality and strategy-proofness, characterize the DA mechanism based on the underlying choice rules.

To that end, fix a school $c\in\mathcal{C}$. A choice rule $\ch_c$ satisfies \textbf{size monotonicity} if $\abs{\ch_c(S')}\leq\abs{\ch_c(S)}$ for every $S,S'\subseteq\mathcal S$ with $S'\subseteq S$. A \textbf{solitary} axiom is a correspondence $\psi_c:2^{\mathcal S}\rightrightarrows 2^{\mathcal S}$ such that for each applicant set $S$, the set $\psi_c(S)\subseteq 2^S$ is the set of all subsets of $S$ that satisfy a particular choice axiom at school $c$. For our purposes, the relevant solitary axioms are non-wastefulness, distributional maximality, and freedom from justified envy. A choice rule $\ch_c$ satisfies a solitary axiom $\psi_c$ if $\ch_c(S)\in\psi_c(S)$ for every $S\subseteq\mathcal S$.

Given a matching $\mu$, define the \textbf{demand} for $c$ at $\mu$ by
\[
D_c(\mu)\coloneqq \{s\in\mathcal S:c\mathrel{R_s}\mu(s)\}.
\]
Since preferences are reflexive in their weak form, $\mu^{-1}(c)\subseteq D_c(\mu)$. We say that $\mu$ satisfies the \textbf{extension} of a solitary axiom $\psi_c$ if $\mu^{-1}(c)\in\psi_c(D_c(\mu))$. 

We now show that the three choice rule axioms in \autoref{thm:exogenouspriority} are solitary and that their extensions coincide with the corresponding matching axioms.

\smallskip
\noindent
\textit{Non-wasteful:} For every $S\subseteq\mathcal S$, define
\[
\psi^{nw}_c(S)\coloneqq
\{S'\subseteq S:\abs{S'}=\min\{q_c,\abs{S}\}\}.
\]
Then a choice rule $\ch_c$ is non-wasteful if and only if $\ch_c(S)\in\psi^{nw}_c(S)$ for every $S\subseteq\mathcal{S}$ (\autoref{ax:nonwasteful}), so this is a solitary axiom.

Recall that a matching $\mu$ must respect a school's capacity, i.e., $\abs{\mu^{-1}(c)}\leq q_c$. A matching $\mu$ satisfies the extension if and only if 
$\abs{\mu^{-1}(c)}=\min\{q_c,\abs{D_c(\mu)}\}$.
This condition trivially holds if $\abs{\mu^{-1}(c)}=q_c$. Because $\mu^{-1}(c)\subseteq D_c(\mu)$, if $\abs{\mu^{-1}(c)}<q_c$, the extension is satisfied if and only if $D_c(\mu)=\mu^{-1}(c)$. Since student preferences are strict, a student $s\notin\mu^{-1}(c)$ belongs to $D_c(\mu)$ if and only if $c\mathrel{P_s}\mu(s)$. Therefore, the extension is equivalent to the matching non-wastefulness condition (\autoref{ax:matching_nonwasteful}): if some student strictly prefers $c$ to their assignment, then \(c\) must have no empty seats.

\medskip
\noindent
\textit{Distributionally maximal:} For every $S\subseteq\mathcal S$, define
\[
\psi^{pd}_c(S)\coloneqq\{S''\subseteq S:\nexists S'\subseteq S \text{ such that }\abs{S'}=\abs{S''}\text{ and }S'\succ_c S''\}.
\]
Then $\ch_c$ is distributionally maximal if and only if $\ch_c(S)\in\psi^{pd}_c(S)$ for every $S\subseteq \mathcal{S}$ (\autoref{ax:diversity}), so this is a solitary axiom.

A matching $\mu$ satisfies the extension if and only if $\mu^{-1}(c)\in\psi^{pd}_c(D_c(\mu))$. Equivalently, there exists no subset $S'\subseteq D_c(\mu)$ with $\abs{S'}=\abs{\mu^{-1}(c)}$ such that $S'\succ_c\mu^{-1}(c)$. This is exactly the matching axiom requiring $\mu$ to be distributionally maximal (\autoref{ax:matching_diversity}).

\smallskip
\noindent
\textit{Free from justified envy:} For every $S\subseteq\mathcal S$, define
\[
\psi^{ne}_c(S)\coloneqq\left\{
S'\subseteq S:
\begin{array}{l}
\text{for every }s'\in S'\text{ and }s''\in S\setminus S',\\
(S'\setminus\{s'\})\cup\{s''\}\succsim_c S'
\text{ implies }s'\mathrel{\priority_c}s''
\end{array}
\right\}.
\]
Then $\ch_c$ is free of $\priority_c$-justified envy if and only if $\ch_c(S)\in\psi^{ne}_c(S)$ for every $S\subseteq\mathcal{S}$ (\autoref{ax:envy}), so this is a solitary axiom.

A matching $\mu$ satisfies the extension if and only if $\mu^{-1}(c)\in\psi^{ne}_c(D_c(\mu))$. Equivalently, for every $s\in\mu^{-1}(c)$ and  $s'\in D_c(\mu)\setminus\mu^{-1}(c)$,
\[
(\mu^{-1}(c)\setminus\{s\})\cup\{s'\}\succsim_c \mu^{-1}(c)\implies s\mathrel{\priority_c}s'.
\]
Because preferences are strict, \(s'\in D_c(\mu)\setminus\mu^{-1}(c)\) is
equivalent to \(c\mathrel{P_{s'}}\mu(s')\). Hence this extension is exactly
the matching condition requiring freedom from \(\priority_c\)-justified envy (\autoref{ax:matching_envy}).

We have thus far shown that for each school $c\in\mathcal{C}$, non-wastefulness, distributional maximality, and freedom from justified envy are all solitary choice axioms. Since $\succsim_c$ satisfies upper-bound and maximizer properties, by \autoref{thm:exogenouspriority}, these solitary axioms uniquely characterize $\ch_c^{\priority_c}$. As $\ch_c^{\priority_c}$ is non-wasteful, it selects $\min\{q_c,\abs{S}\}$ students from each set $S$. Hence, the distributional choice rule also satisfies size monotonicity. Since $\succsim_c$ also satisfies the coherence property, by \autoref{thm:genpi}, $\ch_c^{\priority_c}$ satisfies path independence. 

Moreover, the three solitary axioms extend exactly to the corresponding non-wastefulness, distributional maximality, and freedom from justified envy matching axioms. Thus, applying \citet{Dogan/Imamura/Yenmez:2025} yields that a matching mechanism satisfies the three extended axioms, individual rationality, and strategy-proofness if and only if it is the DA mechanism based on the choice rule profile $(\ch_c^{\priority_c})_{c\in\mathcal C}$.
\end{proof}

\singlespacing
\bibliographystyle{ecta}
\bibliography{matching}

\end{document}